\documentclass[journal]{IEEEtran}
\usepackage{graphicx} 
\usepackage{amsmath}
\usepackage{amsfonts}
\usepackage[ruled, vlined, linesnumbered]{algorithm2e}
\usepackage{color}
\usepackage{url}
\usepackage{tabularray}
\usepackage{subcaption}
\usepackage{amsthm}
\UseTblrLibrary{booktabs}
\newtheorem{definition}{Definition}
\newtheorem{lemma}{Lemma}

\begin{document}

\title{PriFFT: Privacy-preserving Federated Fine-tuning of Large Language Models via Hybrid Secret Sharing}

\author{Zhichao~You, Xuewen~Dong,~\IEEEmembership{Member,~IEEE,} Ke~Cheng, Xutong~Mu, Jiaxuan~Fu, Shiyang~Ma, Qiang~Qu, Yulong~Shen,~\IEEEmembership{Member,~IEEE}
\IEEEcompsocitemizethanks{
\IEEEcompsocthanksitem Zhichao~You, Ke~Cheng, Xutong~Mu, Jiaxuan~Fu, and Yulong~Shen are with the School of Computer Science \& Technology, Xidian University, and are with the Shaanxi Key Laboratory of Network and System Security, Xi'an, China (email: zcyou@stu.xidian.edu.cn, chengke@xidian.edu.cn, \{xtmu, jiaxuanfu\}@stu.xidian.edu.cn, ylshen@mail.xidian.edu.cn).
\IEEEcompsocthanksitem Xuewen~Dong and Shiyang~Ma are with the School of Computer Science \& Technology, Xidian university, and also with
the Shaanxi Key Laboratory of Blockchain and Secure Computing, Xi'an, China  (email: xwdong@xidian.edu.cn, shiyangma@stu.xidian.edu.cn).
\IEEEcompsocthanksitem Qiang Qu is with Shenzhen Institues of Advanced Technology, Chinese Academy of Sciences (e-mail: qiang@siat.ac.cn).
}
}

\maketitle

\begin{abstract}
    Fine-tuning large language models (LLMs) raises privacy concerns due to the risk of exposing sensitive training data. Federated learning (FL) mitigates this risk by keeping training samples on local devices, while facing the following problems in privacy-preserving federated fine-tuning. (i) Recent studies show that adversaries can still infer private information in FL. (ii) LLM parameters are shared publicly during federated fine-tuning, while developers are often reluctant to disclose these parameters, posing further security challenges. (iii) Existing works focus on secure inference of LLMs but do not consider privacy-preserving fine-tuning. Inspired by the above problems, we propose PriFFT, a privacy-preserving federated fine-tuning mechanism, to protect both the model parameters and users' privacy. Due to considerable LLM parameters, we present hybrid secret sharing combining arithmetic secret sharing (ASS) and function secret sharing (FSS) to build secure operations and implement secure layers and activation for privacy-preserving fine-tuning. To improve the efficiency of privacy-preserving federated fine-tuning of LLMs, we optimize several secure computation protocols based on FSS, including reciprocal calculation, tensor products, natural exponentiation, softmax, sigmoid, hyperbolic tangent, and dropout. The hybrid secret sharing enables PriFFT to apply our optimized FSS protocols while combining ASS protocols to support complex computation without extra communication. The optimized protocols reduce execution time up to $62.5\%$ and communication overhead  up to $70.7\%$ compared to existing protocols. Besides, PriFFT reduces execution time and communication overhead in privacy-preserving fine-tuning up to $59.1\%$ and $77.0\%$ without accuracy drop compared to the existing secret sharing methods.
    
\end{abstract}
	

\IEEEdisplaynontitleabstractindextext
\IEEEpeerreviewmaketitle

\section{Introduction}
\IEEEPARstart{L}{arge} language models (LLMs), such as GPT-4~\cite{openai2024gpt4technicalreport}, Llama~\cite{touvron2023llamaopenefficientfoundation}, and BERT~\cite{devlin2019bertpretrainingdeepbidirectional}, are pre-trained on language data to understand the structure, grammar, semantics, and more complex language patterns and concepts of language. Due to the advantages of the LLM parameter scale, LLMs are widely used in search engines, healthcare, finance, and other fields. Fine-tuning pre-trained LLMs according to downstream tasks allows LLMs to achieve higher accuracy and improved performance in specific domains. However, the training samples of downstream tasks contain sensitive information, causing privacy leakage concerns when fine-tuning LLMs in downstream tasks~\cite{carlini2021extractingtrainingdatalarge}. For example, fine-tuning LLMs for the medical domain involves direct access to patient disease information and medical records. Therefore, it is essential to prevent sensitive data leakage while enhancing model performance in fine-tuning LLMs.

Federated Learning (FL)~\cite{kone2017federated,zhang2024federated,xiao2024fedrma} ensures the security of individual data during model training, in which models are trained locally on the data owner's device, while only sharing model parameters. The training data remains local, and the model developer never gets the original training data, which provides privacy protection for sensitive sample data. Therefore, fine-tuning LLMs with FL, i.e., federated fine-tuning, provides privacy preservation for training data in downstream tasks~\cite{chen2023federatedlargelanguagemodel}.
Figure~\ref{fig-fine-tuning-FL-process} presents the framework of fine-tuning LLMs with FL where training samples remain local and only model updates are exchanged in the training process.
Federated fine-tuning freezes most of the model's parameters and fine-tunes LLMs by adjusting a few selected parameters, alleviating the issue of limited computational resources for clients in FL.

\begin{figure}
    \centering
    \includegraphics[width=0.95\linewidth]{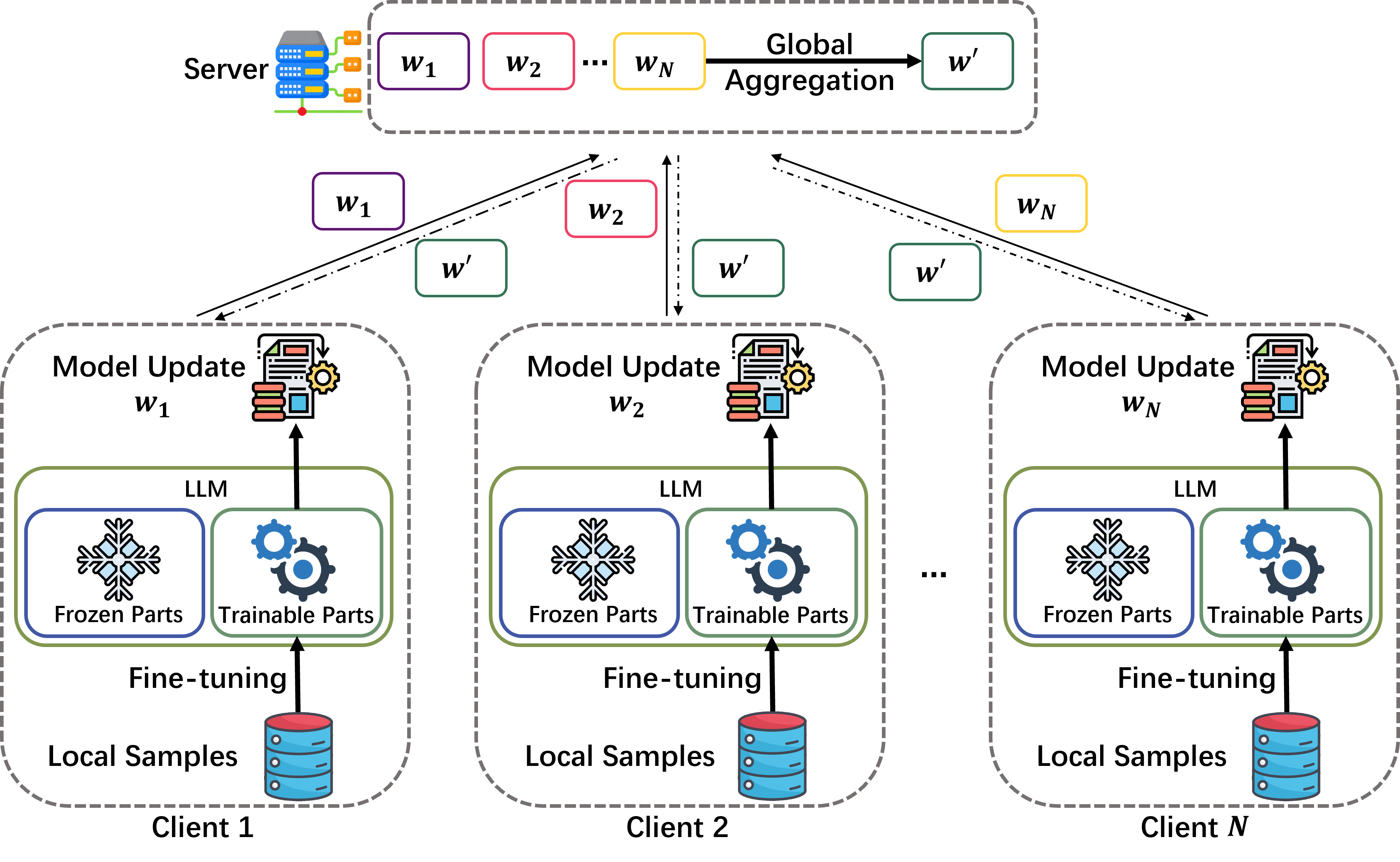}
    \caption{A simple framework of federated fine-tuning.}
    \label{fig-fine-tuning-FL-process}
\end{figure}

\textbf{Motivations.} Federated fine-tuning faces two security threats from privacy leakage of model updates and parameter exposure of trained models. (i) Although the server cannot directly access training samples, attacks against FL show that \textbf{the server can still infer training sample information through model updates~\cite{MOTHUKURI2021619}, which is the first security threat.} For example, model inversion attacks~\cite{NEURIPS2019_60a6c400} reconstruct training samples through model gradients uploaded by clients, presenting the privacy risks of uploading model updates generated by local samples. (ii) Granting \textbf{clients direct access to trained model parameters in FL raises concerns about intellectual property, which is the second security threat.} Since model developers invest substantial computational resources, money, and effort in model training, they are cautious about disclosing trained model parameters to protect intellectual property and avoid commercial competition risks~\cite{he2022protecting}.

\textbf{Challenges.} Existing privacy-preserving FL~\cite{Yin2021ComprehensiveSurvey} mitigates the above privacy threats by implementing secure aggregation via secure multi-party computation (SMPC)~\cite{goldreich1998secure}. The existing privacy-preserving mechanisms need to solve three challenges in federated fine-tuning of LLMs. Clients train the global model with local samples and generate model updates in plaintext. In the secure aggregation, the model updates are encrypted or perturbed before uploading. The server performs the model aggregation on the protected model updates with SMPC and generates updated global models. In the aggregation stage, the server cannot access the model updates, preventing the server from performing inference attacks on model updates and protecting clients' sample security. However, \textbf{the secure aggregation still requires model parameters to be exposed to clients, which cannot provide privacy preservation for model parameters, which is the first challenge}.

Existing privacy-preserving mechanisms for LLMs apply secret sharing during inference to protect both the model's parameters and the user's privacy~\cite{Crypten,SHAFT,Kanav2022LLAMA}.
However, these mechanisms \textbf{overlook privacy-preserving fine-tuning of LLMs and are not directly applicable to privacy-preserving federated fine-tuning}, as fine-tuning involves more complex secure computations compared to the inference, \textbf{which is the second challenge.}
On the other hand, existing works rely on arithmetic secret sharing (ASS)~\cite{Crypten} or function secret sharing (FSS)~\cite{Kanav2022LLAMA}, \textbf{lacking a combined application of ASS and FSS, which is the third challenge.}
ASS-based multi-party secure computation libraries (such as ABY2~\cite{ABY} and CrypTen~\cite{Crypten}) implement numerous foundational secure operations, thus eliminating the need to build privacy-preserving federated learning systems entirely from the ground up.
In the meantime, FSS offers an alternative means of reducing the overhead in secure computation protocols.
An advantageous strategy entails combining ASS and FSS into hybrid protocols with minimal additional overhead, thereby reducing resource consumption within the privacy-preserving system.

In this paper, we propose PriFFT, a privacy-preserving federated fine-tuning framework via hybrid secret sharing to protect both model updates and parameters.
We propose efficient computation protocols based on FSS to decrease the communication overhead of bottleneck protocols in federated fine-tuning by minimizing the number of communication rounds and the volume of data exchanged.
The hybrid secret sharing enables PriFFT to apply our optimized FSS protocols while combining ASS protocols to support complex computation without extra communication.
Meanwhile, PriFFT leverages GPU acceleration to further reduce execution time.
Compared to the fine-tuning under plaintext data, the evaluation shows that PriFFT causes a minor drop in model accuracy while providing privacy-preserving for both model parameters and updates. Meanwhile, PriFFT significantly reduces communication and computation overhead compared to implementations based on existing privacy-preserving mechanisms.
The contributions of this paper are summarized as follows:
\begin{itemize}
    \item \textbf{Privacy-preserving federated fine-tuning of LLMs.} PriFFT is the first privacy-preserving FL mechanism for fine-tuning pre-trained LLMs according to downstream tasks. PriFFT enables federated fine-tuning of LLMs while protecting local samples and model updates of clients, as well as the parameters of trained model.

    \item \textbf{Efficient implementation of bottleneck protocols.} Dealing with the substantial secure computations in fine-tuning, we propose efficient computation protocols, including reciprocal calculation, tensor product, natural exponential, softmax, hyperbolic tangent, and dropout functions based on FSS.
    The optimized protocols reduce execution time up to $62.5\%$ and communication overhead  up to $70.7\%$ compared to existing secure computation protocols.

    \item \textbf{Fine-tuning based on hybrid secret sharing.} We present the hybrid secret sharing combining ASS and FSS to construct PriFFT. The hybrid secret sharing mechanism enables PriFFT to implement the proposed FSS protocols, minimizing the bottleneck secure operation's overhead and facilitating complex secure computations based on ASS. The share conversion of hybrid secret sharing allows PriFFT to utilize both ASS and FSS protocols in privacy-preserving federated fine-tuning without incurring additional communication overhead.
    
    \item \textbf{Theoretical and experimental validity proof.}  We ensure the security and efficiency of PriFFT through theoretical analysis. Evaluation results on real standard datasets demonstrate that PriFFT realizes comparable model accuracy to plaintext training while protecting both model parameters and updates. With the same model accuracy of the ABY2-based implementation, PriFFT reduces execution time and communication overhead in privacy-preserving fine-tuning up to $59.1\%$ and $77.0\%$ without accuracy drop compared to the existing secret sharing methods.
\end{itemize}

The rest of this paper is organized as follows. Section~\ref{section-related-work} summarizes the related work. Section~\ref{section-preliminary} and Section~\ref{section-system-overview} provide the preliminaries and a system overview, respectively. Section~\ref{section-mechanism} presents the design of PriFFT. Section~\ref{section-analysis} discusses the complexity and provides the security proofs of PriFFT. Section~\ref{section-evaluation} evaluates PriFFT through comprehensive experiments. Finally, Section~\ref{section-conclusion} concludes the paper.

\section{Related Work}
\label{section-related-work}

\subsection{Federated Fine-tuning of LLMs}
\label{preliminary-fine-tuning}

Several existing works propose FL mechanisms to fine-tune LLMs and protect clients' sample data. Qin et al. introduce zeroth-order optimization (ZOO) to realize full-parameter tuning of LLMs~\cite{qin2024federated}. ZOO may struggle to achieve the same accuracy and efficiency as backpropagation when performing full-parameter fine-tuning of LLMs, since ZOO cannot precisely adjust model parameters by gradients.
FedAdapter~\cite{cai2023FedAdapter} implements adapter-based fine-tuning of LLMs with FL, which embeds adapters into LLMs and only optimizes adapter parameters according to downstream tasks. FedPETuning~\cite{zhang-etal-2023-fedpetuning} develops a benchmark for four LLM fine-tuning methods, including adapter tuning, prefix tuning, and LoRA. FedPrompt~\cite{zhao2023FedPrompt} embeds soft prompts into global models and only adjusts the soft prompt parameters during the training process to realize communication-efficient fine-tuning of LLMs.

Exiting works on fine-tuning LLMs with FL do not provide privacy preservation for both model parameters and clients' model updates. As clients fine-tune models on local devices in FL, the server must distribute model parameters to clients, leading to the exposure of these trained parameters. Model developers are cautious about disclosing trained model parameters due to potential intellectual property issues and commercial competition. Besides, clients fine-tune models and generate model updates based on the sample data specific to downstream tasks. The server can directly access the model updates in model aggregation and reconstruct training samples through model inversion attacks~\cite{MOTHUKURI2021619} using model updates, presenting the privacy leakage risk when uploading model updates in plaintext.


\subsection{Privacy-preserving Federated Learning}
Existing privacy-preserving FL mechanisms provide protection for model updates through different secure strategies, including differential privacy (DP), encrypted aggregation, and secret sharing.
DP-based FL mechanisms protect model updates and parameters by adding noise to the sensitive data~\cite{Wei2020DP}.
Clients perturb model updates with noise before uploading the model updates. Similarly, the server perturbs model parameters before distributing the parameters. Clients and the server receive perturbed sensitive data with noise, achieving protection for model updates and parameters. However, noise in perturbed updates and parameters brings model performance drops. When fine-tuning LLMs with extensive parameters, the perturbation noise in model updates and parameters significantly impacts model accuracy.

In encrypted aggregation, model updates are transmitted in ciphertext, preventing the server from accessing the sensitive data in plaintext. BatchCrypt~\cite{zhang2020batchcrypt} employs homomorphic encryption (HE) to implement secure aggregation, in which gradients are encrypted via HE and gradient aggregation is realized through secure computation. Tamer et al. present a privacy-preserving FL based on symmetric encryption, addressing the clients' dropout problem and enabling clients to validate the aggregated models~\cite{eltaras2023efficient}. SAFELearn~\cite{SAFELearn} presents a communication-efficient and privacy-preserving FL framework that can be instantiated with SMPC or HE to defend against model inversion attacks on model updates. Abbas et al. present an HE-based privacy-preserving FL against model poisoning attacks with an internal auditor and Byzantine-tolerant gradient aggregation~\cite{Yazdinejad2024roubust}. Few existing works of encryption-based FL consider the privacy leakage issues related to model parameters. When model parameters are encrypted, the models' local training must be done through secure computation. Fine-tuning LLMs with extensive encrypted model parameters imposes a significant computational load on clients.

The key idea of secret sharing-based FL mechanisms is to split model updates into multiple shares. Aggregation results are computed on shares, preventing the server from accessing clients' gradients during the aggregation. 
FastSecAgg~\cite{kadhe2020fastsecagg} combines multi-secret sharing and fast Fourier transform to realize a trade-off between the number of secrets, privacy threshold, and dropout tolerance. ELSA~\cite{Mayank2022ELSA} introduces two servers in FL and splits clients' gradients into two shares. Each server receives a share of clients' gradients and performs gradient aggregation on received shares.
VerSA~\cite{Hahn2023VerSA} utilizes key agreement protocols to generate seeds for pseudorandom number generators, which are then used to generate masks for model updates. The aggregated results on masked model updates can offset the random values of the masks. Additionally, the secret sharing mechanism in VerSA addresses the issue of client dropout. Secret sharing-based FL mechanisms for fine-tuning of LLMs present significant communication overhead challenges, as the non-linear computations in fine-tuning of LLMs require extensive communication between clients and the server. Considering the limited communication capabilities of clients' devices, clients may not be able to bear the communication consumption of FL mechanisms that directly apply secret sharing for privacy preservation.


The comparison of the existing works is briefly concluded in TABLE~\ref{table: comparison reference}. This paper presents PriFFT to protect model parameters and updates with provable security while reducing computation and communication overhead in federated fine-tuning of LLMs.

\begin{table*}
\centering
    \caption{Comparison of relative works.}
    \label{table: comparison reference}
    \begin{tblr}{rows={c,m}, column{4,5,7}={2cm}, hline{2-Y}, column{2,3}={1.8cm}, column{6}={2.2cm}}
        \toprule
        \SetCell[r=2]{c,m} Mechanism & \SetCell[r=2]{c,m} Federated Learning & \SetCell[r=2]{c,m} Fine-tuning of LLMs & \SetCell[r=2]{c,m} Sample Protection & \SetCell[r=2]{c,m} Model Update Protection & \SetCell[r=2]{c,m} Model Parameter Protection & \SetCell[r=2]{c,m} SMPC Technique \\
        & & & & & & \\
        ZOO~\cite{qin2024federated} & $\surd$ & $\surd$ & $\surd$ & $\times$ & $\times$ & $\times$ \\
        FedAdapter~\cite{cai2023FedAdapter} & $\surd$ & $\surd$ & $\surd$ & $\times$ & $\times$ & $\times$\\
        FedPETuning~\cite{zhang-etal-2023-fedpetuning} &  $\surd$ & $\surd$ & $\surd$ & $\times$ & $\times$ & $\times$\\
        FedPrompt~\cite{zhao2023FedPrompt} & $\surd$ & $\surd$ & $\surd$ & $\times$ & $\times$ & $\times$\\
        Wei's work~\cite{Wei2020DP} & $\surd$ & $\times$ & $\surd$ & $\surd$ & $\surd$ & DP \\
        BatchCrypt~\cite{zhang2020batchcrypt} & $\surd$ & $\times$ & $\surd$ & $\surd$ & $\times$ & HE \\
        LLAMA~\cite{Kanav2022LLAMA}  & $\times$ & $\times$ & $\surd$ & $\surd$ & $\surd$ & FSS\\
        CrypTen~\cite{Crypten} & $\times$ & $\times$ & $\surd$ & $\surd$ & $\surd$ & ASS\\
        SHAFT~\cite{SHAFT} & $\times$ & $\times$ & $\surd$ & $\surd$ & $\surd$ & ASS\\
	PriFFT (our solution) & $\surd$ & $\surd$ & $\surd$ & $\surd$ & $\surd$ & ASS+FSS \\
        \bottomrule
    \end{tblr}
\end{table*}

\section{Preliminary}
\label{section-preliminary}

\subsection{Federated Learning}
\label{preliminary-fl}
In the training initialization phase, the server designs model structure $f$ with initialized parameters $\omega_0$, which are distributed to clients~\cite{kone2017federated}. In the $t$-th training round, clients download the latest model parameters $\omega_t$ and update local models. Assume that $N$ clients participate in FL training, the local training data of client $i \in \left[1,\cdots,N\right]$ are samples $x_i$ and labels $y_i$. The inference results are $f(x_i;w_t)$. Given the loss function $L$, the loss of client $i$ in local training is $L(f(x_i;w_t), y_i)$, and gradients are $g_i = \partial L(f(x_i;w_t), y_i) / \partial \omega_t$.

Clients upload gradients $g_i$ to the server in the aggregation phase of the $t$-th training round. 
The server averages gradients to get aggregated gradients as $\Bar{g} = \sum_{i=1,\cdots, N} g_i$ and optimizes the parameters of the global models by aggregated gradients $\Bar{g}$.
The global model's updated parameters $\omega_{t+1}$ are distributed to clients for the subsequent round of local training. The training process is terminated once a specified number of training rounds has been completed.


\subsection{Arithmetic Secret Sharing}
\label{preliminary-ass}
We consider that secure computations are performed in the ring $\mathbb{Z}_{2^l}$, and all computations are modulo $2^l$, where $l$ is the bitwidth of the ring~\cite{patra2021aby2}.
We map a real number $x_r \in \mathbb{R}$ to the item of the ring $x \in \mathbb{Z}_{2^l}$ by $x = \left \lfloor x_r \cdot 2^s \right \rfloor \; \mathrm{mod} \; 2^l$ where $s$ is a scale factor.
The plaintext values are additively shared, and each party cannot learn information from the plaintext.
Due to the space limitation, the sharing semantics and operations of ASS are given in the Appendix~\ref{appendix-ASS}.

\subsection{Function Secret Sharing}
\label{preliminary-FSS}
Due to communication efficiency, FSS has been applied in various areas to implement privacy-preserving mechanisms~\cite{GoCrowd,FSS-DBSCAN,eGrass,fu2024private}.
We consider a 2-party FSS scheme in this paper.
Given a function family $\mathcal{F}$, FSS splits a function $f: \mathbb{G}^{\mathrm{in}} \rightarrow \mathbb{G}^{\mathrm{out}}$ into two additive shares $\{f_0, f_1\}$ such that $f_0 \left( x\right) + f_1 \left( x\right) = f \left( x\right)$ where $x \in \mathbb{G}^{\mathrm{in}}$ and $f \left( x\right) \in \mathbb{G}^{\mathrm{out}}$.
The security property of FSS requires that each share $f_b$ hides $f$.
The formal definition of FSS and the key generation in the offline stage are given in Appendix~\ref{appendix-fss}.

The evaluation $\mathrm{Eval}(\sigma, k_{\sigma}, x)$ of FSS involves public input $x$, which is inconsistent in protecting sensitive data in privacy-preserving FL.
Boyle et al. present secure evaluations in FSS based on offset functions~\cite{boyle2019secure}.
The key idea is that for each gate $g: \mathbb{G}^{\mathrm{in}} \rightarrow \mathbb{G}^{\mathrm{out}}$ in the computation circuit, the input wire $w_i$ and output wire $w_j$ are masked by random offsets $r_i$ and $r_j$. Specifically, offset functions $\hat{\mathcal{G}}$ includes functions of the form $g^{ \left[ r^{\mathrm{in}}, r^{\mathrm{out}}
 \right] } (x) = g(x - r^{\mathrm{in}}) + r^{\mathrm{out}}$.
Each party evaluates FSS shares on the common masked input $w_i + r_i$ and obtains additive shares of the masked output $w_j + r_j$ to get the correct shares. 
The formal definition of FSS with the offset function is given in Appendix~\ref{appendix-fss}.
In this paper, $\hat{x}$ refers to the fact that $x$ is masked by an offset $r$, i.e., $\hat{x} = x + r$. The notation $\left\langle \hat{x} \right\rangle_b$ refers to an additive share of $\hat{x}$ held by party $b \in {0,1}$.



\subsection{Share Conversion}
\label{preliminary-share-conversion}
The proposed privacy-preserving federated fine-tuning involves hybrid secure computing protocols based on ASS and FSS.
The values are shared by ASS ($\left \langle x \right \rangle$) as discussed in Section~\ref{preliminary-ass} by default.
When executing FSS-based protocols, the input of secure computing is $\left \langle \hat{x} \right \rangle$ which is then used to derive $\hat{x}$. The conversion process is defined as follows:

Conversion from $\left \langle x \right \rangle$ to $\left \langle \hat{x} \right \rangle$ is implemented by having each party perform $\left \langle \hat{x} \right \rangle = \left \langle x \right \rangle + \left \langle r^{\mathrm{in}} \right \rangle$.
FSS-based protocols first perform the conversion to obtain $\left \langle \hat{x} \right \rangle$.
Some FSS-based protocols invoke $\left \langle r^{\mathrm{in}} \right \rangle$ to obtain computation results.
To ensure that each party cannot infer information about $x$ from $\left \langle \hat{x} \right \rangle_b$, we split $r^{\mathrm{in}}$ twice in the input conversion and offline stages.

\textbf{Conversion from $\left \langle \hat{x} \right \rangle$ to $\left \langle x \right \rangle$} is required when an ASS-based protocol follows an FSS-based protocol.
The conversion is implemented by omitting $\left \langle r^{\mathrm{out}} \right \rangle$ in the output of an FSS-based protocol.
Without $\left \langle r^{\mathrm{out}} \right \rangle$, the output is converted to additive shares, while the key size is reduced $l$-bits due to the omission.


\section{System Overview and Threat Model}
\label{section-system-overview}

\subsection{Threat Model}
\textbf{Adversary's identity.} In this paper, we consider semi-honest servers and clients. The server and clients adhere to computation protocols while attempting to infer as much private information as possible from the calculation results.

\textbf{Adversary's goal.} The server would infer sensitive information about the training samples through attacks against clients' model updates, such as model inversion, label inference, and membership inference attacks. Clients also infer model parameters during the fine-tuning process.

\textbf{Adversary's ability and knowledge.} The server possesses the model parameters from each training round, while clients retain local training data. Both the server and clients can directly access intermediate calculation results.

\textbf{Defense goals.} Since multiple attacks against FL can infer sensitive information from model updates, one of the defense goals is to protect uploaded model updates, i.e., the server cannot directly access the original model updates. Another defense goal is to guarantee that clients cannot obtain the model parameters by intermediate calculation results.

\subsection{System Model}
\begin{figure}
    \centering
    \includegraphics[width=0.95\linewidth]{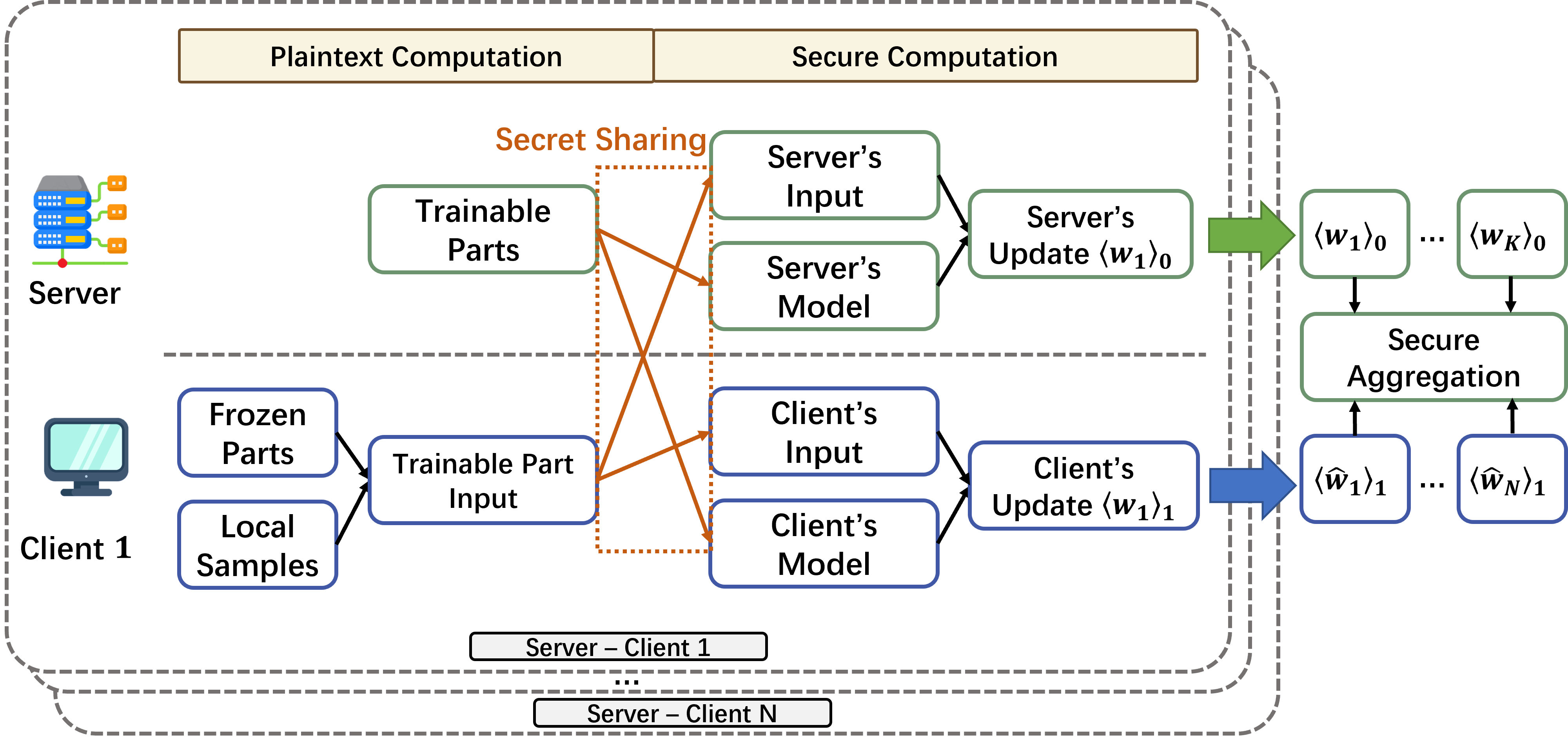}
    \caption{A simple illustration of PriFFT framework.}
    \label{fig-fl-ft-framework}
\end{figure}

Fig.~\ref{fig-fl-ft-framework} presents the framework of PriFFT, and we provide a more detailed implementation overview in Section~\ref{section-implementation-overview}.
There are $N$ clients participating in FL training. For each client $i$, the training data is presented as $(x_i, y_i)$, where $x_i$ and $y_i$ are training samples and labels. The model structure and parameters are $f$ and $\omega$, respectively.
Parameters are divided into the frozen part $\omega_{\mathbf{F}}$ and the trainable part $\omega_{\mathbf{T}}$, where $\omega=\omega_{\mathbf{F}} \cup \omega_{\mathbf{T}}$. 
The frozen part $\omega_{\mathbf{F}}$ remains constant, and the trainable part $\omega_{\mathbf{T}}$ is adjusted by clients through local training.

The trainable model parameters $\omega_{\mathbf{T}}$ are transmitted to clients in the form of shares. 
The number of training rounds is $T$. In the $t$-th ($t \in \left[ 1,\cdots,T \right]$) training round, the server splits the latest trainable parameters $\omega_{\mathbf{T}, t}$ into $\left \langle \omega_{\mathbf{T}, t} \right \rangle_b$ ($b\in \{0, 1\}$) where $\omega_{\mathbf{T}, t} = \left \langle \omega_{\mathbf{T}, t}\right \rangle_0 + \left 
\langle \omega_{\mathbf{T}, t}\right \rangle_1$. 
Similarly, client $i$ randomly samples and splits a batch of sample $(x_i, y_i)$ such that $x_i = \left \langle x_i\right \rangle_0 + \left \langle x_i\right \rangle_1$ and $y_i = \left \langle y_i\right \rangle_0 + \left \langle y_i\right \rangle_1$. 
The server and client $i$ exchange shares of model parameters and sample data so that the server gets $\{ \left \langle \omega_{\mathbf{T}, t} \right \rangle_0, \left \langle x_i\right \rangle_0, \left \langle y_i\right \rangle_0 \}$ and client $i$ gets $\{ \left \langle \omega_{\mathbf{T}, t} \right \rangle_1, \left \langle x_i\right \rangle_1, \left \langle y_i\right \rangle_1 \}$.
Model parameters in the $t$-training round are $\omega_t = \omega_\mathbf{F} \cup \omega_{\mathbf{T}, t}$. 
Since PriFFT splits $\omega_{\mathbf{T}, t}$ into two shares, the protected model parameters are represented by $\hat{\omega}_{b,t} = \omega_\mathbf{F} \cup \left \langle \omega_{\mathbf{T}, t} \right \rangle_b$ where $b \in \{0,1\}$.

PriFFT transfers $f_\mathbf{T}$ to $\hat{f}_\mathbf{T}$ by our designed secure computation protocols such that $f(x_i;\omega_t) = \hat{f}(\left \langle x_i \right \rangle_0;\hat{\omega}_{0,t}) + \hat{f}(\left \langle x_i \right \rangle_1; \hat{\omega}_{1,t})$. 
Similarly, PriFFT transfers the loss function $L$ to $\hat{L}$ and presents secure gradient computation method such that $L(f(x_i;w_t), y_i) = \hat{L}(\hat{f}(\left \langle x_i\right \rangle_0;\hat{\omega}_{0,t}), \left \langle y_i\right \rangle_0) + \hat{L}(\hat{f}(\left \langle x_i\right \rangle_1;\hat{\omega}_{1,t}), \left \langle y_i\right \rangle_1)$ and
\begin{equation}
\label{equation-shares-of-gradients}
\frac{\partial L(f(x_i;w_t), y_i)}{\partial \omega_{\mathbf{T}, t}} = \sum_{b \in \{0,1\}} \frac{\hat{L}(\hat{f}(\left \langle x_i\right \rangle_b;\hat{\omega}_{b,t}), \left \langle y_i\right \rangle_b)}{\partial \left \langle 
\omega_{\mathbf{T}, t} \right \rangle_b}.
\end{equation}
We use $\left \langle g_{i,t}\right \rangle_b$ to represent the protected gradients of trainable parameters in the right term of Equation~\eqref{equation-shares-of-gradients}. 
In other words, $\left \langle g_{i,t}\right \rangle_b$ is the share of the trainable parameter gradients generated by client $i$'s samples in the $t$-th training round. 
Similarly, the left term of Equation~\eqref{equation-shares-of-gradients} is presented by $g_{i,t}$.
In the aggregation phase of the $t$-th training round, the server needs to calculate the aggregated gradients by
\begin{equation}
\bar{g}_t = \frac{1}{N} \sum_{i=1}^{N} g_{i,t} = \frac{1}{N} \sum_{i=1}^{N} \left( \left \langle g_{i,t}\right \rangle_0 + \left \langle g_{i,t}\right \rangle_1 \right).
\end{equation}
The server can directly get $\left \langle g_{i,t}\right \rangle_1$ by $\sum_{b\in \{0,1\}} \left \langle g_{i,t}\right \rangle_b$ if client $i$ uploads $\left \langle g_{i,t}\right \rangle_1$ to the server, leading to privacy leakage of clients' training samples due to reconstruction attacks against gradients. 
We apply the secure aggregation~\cite{Hahn2023VerSA} to protect $\left \langle g_{i,t}\right \rangle_b$. Specifically, client $i$ masks shared gradients $\left \langle g_{i,t}\right \rangle_1$ to double masked gradients $[\![g_{i,t}]\!]_1$ such that
\begin{align}
    \left\{\begin{matrix}
g_{i,t} \neq \left \langle g_{i,t}\right \rangle_0 + [\![g_{i,t}]\!]_1;
 \\
\bar{g}_t = \frac{1}{N} \sum_{i=1}^{N} g_{i,t} = \frac{1}{N} \sum_{i=1}^{N} \left( \left \langle g_{i,t}\right \rangle_0 + [\![g_{i,t}]\!]_1 \right).
\end{matrix}\right.
\end{align}
In other words, the server cannot restore $g_{i,t}$ by $[\![g_{i,t}]\!]_1$ but still get the correct aggregated gradients $\bar{g}_t$ by the sum of gradients from the server and clients.

The process of the $t$-th training round can be summarized as: \textbf{(\romannumeral 1)} the server and client $i$ send $\hat{\omega}_{1,t}$ and $\{\left \langle x_i\right \rangle_0, \left \langle y_i\right \rangle_0\}$ to each other, respectively; \textbf{(\romannumeral 2)} the server and client $i$ compute $\left \langle g_{i,t}\right \rangle_b$ via presented secure computation protocols; \textbf{(\romannumeral 3)} client $i$ masks $\left \langle g_{i,t}\right \rangle_1$ to $[\![g_{i,t}]\!]_1$ and uploads $[\![g_{i,t}]\!]_1$ to the server; \textbf{(\romannumeral 4)} the server calculates aggregated gradients $\bar{g}_t$ and optimizes the model parameters. The above process occurs in parallel between the server and each client.

\section{PriFFT: Mechanism and Protocol Design}
\label{section-mechanism}
This section delves into the implementation details of PriFFT as well as the optimization for the related protocols used in PriFFT.
Section~\ref{section-implementation-overview} provides an implementation overview of PriFFT and demonstrates the application of the proposed protocols for constructing a privacy-preserving federated fine-tuning of LLMs from a high-level viewpoint.
Section~\ref{section-hybrid-shares} xxx.
Section~\ref{section-protocol-design} presents the optimized computation protocols used in PriFFT.
The detailed construction of the protocol includes numerous pseudo-code algorithms for both the offline and online phases.

\subsection{Implementation Overview}
\label{section-implementation-overview}
\begin{figure}
    \centering
    \includegraphics[width=0.98\linewidth]{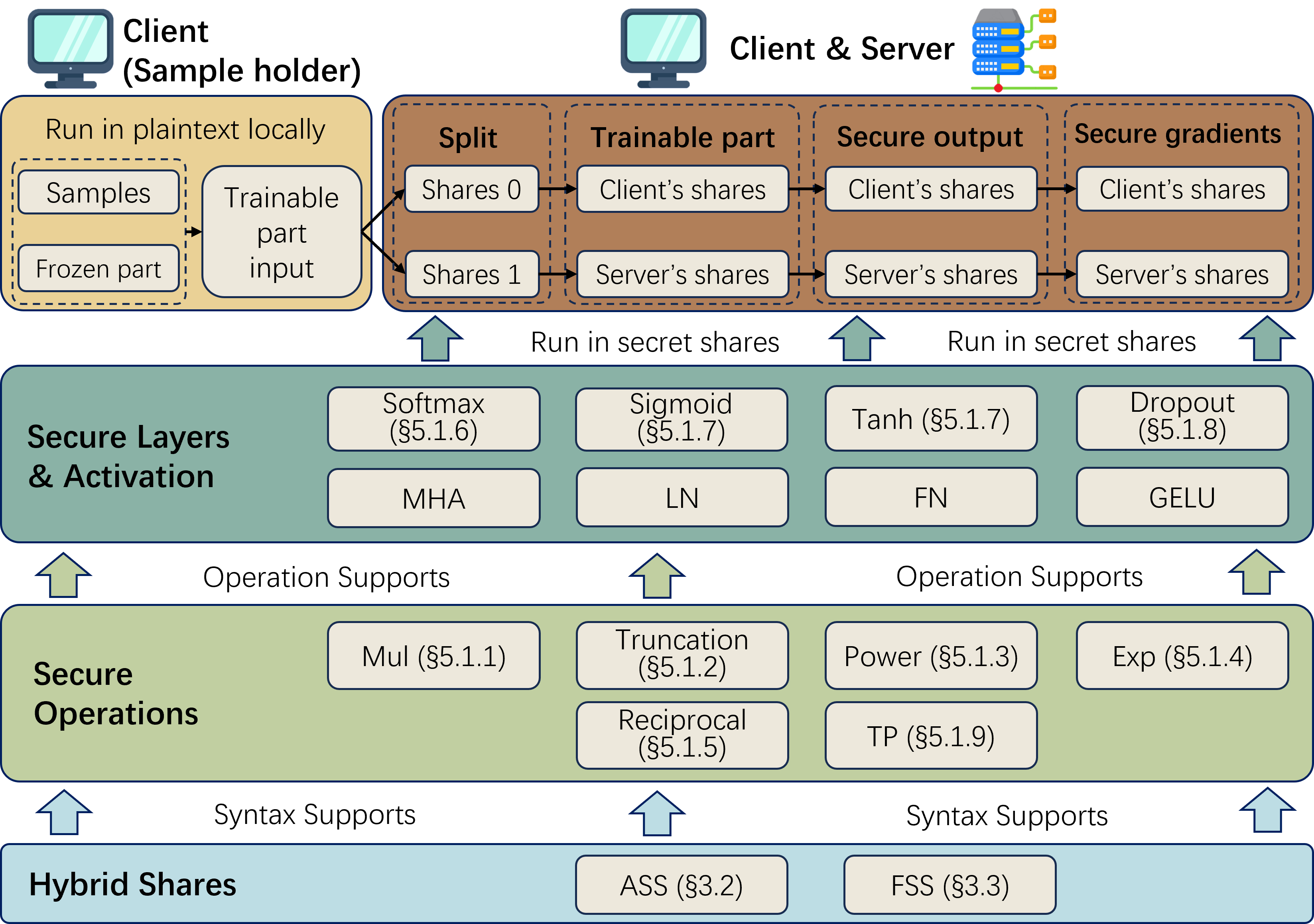}
    \caption{The design and implementation of PriFFT from a high-level perspective.}
    \label{fig-implementation}
\end{figure}

Fig.~\ref{fig-implementation} provides the design and implementation of PriFFT from a high-level perspective, which elaborates on and augments the framework presented in Fig.~\ref{fig-fl-ft-framework}, illustrating the implementation of PriFFT via the proposed protocols.
Fig.~\ref{fig-fl-ft-framework} aligns with the top process depicted in Fig.~\ref{fig-implementation}, where clients locally generate inputs for the trainable part models in plaintext and split these inputs through secret sharing.
Before the training process or each training iteration, the model's trainable parameters are split in advance by secret sharing.
Therefore, the training process that generates secure gradients used in updating parameters in fine-tuning is run in secret shares, which protects both the model parameters and clients' samples.

PriFFT achieves privacy-preserving federated fine-tuning by integrating secure neural network layers and activation functions, illustrated as the second layer in Fig.~\ref{fig-implementation}.
We optimize several secure activation functions like Softmax, Sigmoid, Tanh, along with dropout functions, whereas other secure model structures such as multi-head attention (MHA), layer normalization (LN), and feed-forward networks (FN) are executable through foundational secure protocols.
The third layer in Fig.~\ref{fig-implementation} represents the proposed secure operations, facilitating the execution of secure layers and activation functions.
PriFFT optimizes secure operations, layers, and activation functions, resulting in reduced computing time and communication overhead when compared to existing secret sharing-based methods.
The concept of hybrid shares proposed in Section~\ref{section-hybrid-shares} provides the syntax supports for higher-level protocols.

\subsection{Hybrid Shares}
\label{section-hybrid-shares}

The share conversion is given in Section~\ref{preliminary-share-conversion}. 
We extend the concept of share conversion to hybrid secret sharing, aiming to enhance the efficiency of secure computation protocols by minimizing overhead.
The hybrid secret sharing is motivated by the following factors.
While some protocols based on FSS (as mentioned in Sections~\ref{section-exp}-\ref{section-tensor-prodcut}) are capable of minimizing computational overhead compared to ASS, there exist protocols that result in similar resource consumption (such as the multiplication in Section~\ref{appendix-mul}).
As a common secret-sharing technique, the ASS-based secret-sharing library currently available (like CrypTen~\cite{Crypten}) executes numerous operations necessary for machine learning.
Constructing a secure model's predictions and training system in machine learning via FSS entails reworking numerous foundational protocols and addressing the relevant engineering challenges.
One more efficient approach involves crafting optimized protocols based on FSS utilizing hybrid shares to substitute the ASS protocols in the existing secret-sharing library, thereby minimizing the overhead of secure fine-tuning of LLMs.

\begin{figure}
    \centering
    \includegraphics[width=0.99\linewidth]{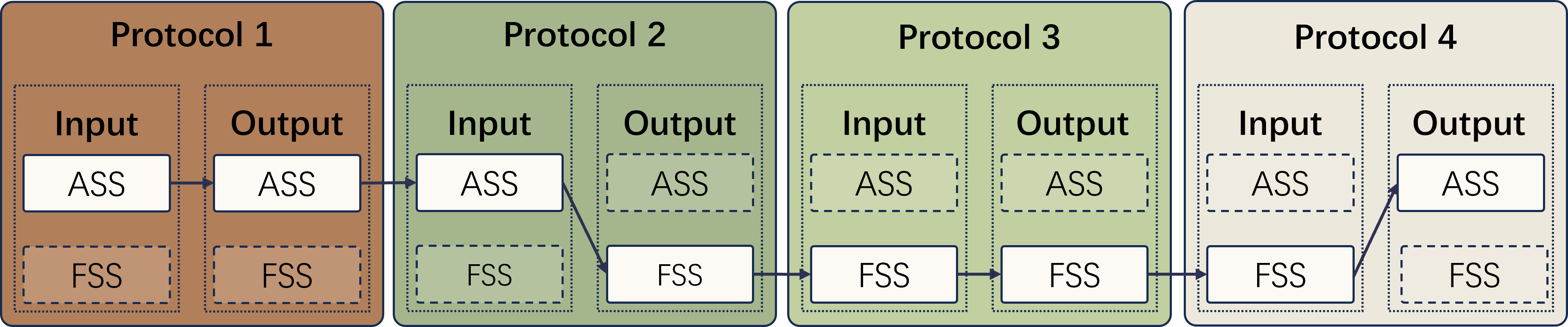}
    \caption{Illustration of Hybrid Shares in Combined Protocols.}
    \label{fig-hybrid-shares}
\end{figure}

We provide a brief illustration of the concept of hybrid shares using the examples presented in Fig.~\ref{fig-hybrid-shares}.
Protocols 1 and 3 have inputs and outputs of identical types without requiring share conversion, and both are standard ASS or FSS protocols.
During the implementation of protocol 2, if protocol 3 can be enhanced via FSS, then protocol 2 will produce FSS shares as its output.
The aforementioned transformation can be carried out using $\left\langle \hat{x} \right\rangle = \left\langle x \right\rangle + \left\langle r^{\mathrm{in}} \right\rangle$, where $\left\langle x \right\rangle$, $\left\langle \hat{x} \right\rangle$, and $\left\langle r^{\mathrm{in}} \right\rangle$ represent, respectively, the ASS output before the share conversion, the FSS output after the share conversion, and the input randomness associated with FSS protocols (protocol 3).
In a similar manner, if the subsequent protocol in protocol 4 receives ASS shares as input, protocol 4 is also capable of transforming FSS shares output into ASS shares.
The implementation method is to omit $\left\langle r^{\mathrm{out}} \right\rangle$ in the FSS output.

By integrating share conversion into secure computing protocols, hybrid shares can be incorporated into the mechanism to enable fully federated fine-tuning of LLMs.
Because share conversion incurs minimal computational overhead and doesn't increase communication consumption, the FSS protocols can seamlessly substitute the ASS protocols within the current mechanisms.

\subsection{Protocol Design}
\label{section-protocol-design}
Based on the concept of hybrid shares, we introduce the optimized secure computing protocols developed by PriFFT in this section.
Section~\ref{section-protocol-design} only presents the protocol design, while the detailed pseudo-code and further discussion are located in the appendix for interested readers due to space constraints.

\subsubsection{Power Functions}
\label{section-power-functions}
Since several computations in PriFFT involve the secure square function, we present a secure square function based on offset functions of FSS to reduce the communication overhead.
The straightforward way to implement the square function is using ASS multiplication given in Section~\ref{preliminary-ass}. 
The communication overhead of computing $x^2$ for a value $x \in \mathbb{Z}_{2^l}$ is $5l$ bits.
We optimize the protocol based on FSS and reduce communication overhead from $5l$ bits to $3l$ bits.
We leave the construction detail of the secure square function in APPENDIX~\ref{appendix-power}.
Similar to the construction of the secure square function, we present the secure multiplication based on FSS to implement multiplications with FSS, while the construction detail is given in Appendix~\ref{appendix-mul}.
Besides, it is straightforward to construct the secure power function, which is similar to the construction of the secure square function.
Due to the space limitation, the detailed construction of the secure power function is given in APPENDIX~\ref{appendix-power}.

\subsubsection{Natural Exponential Function}
\label{section-exp}


We consider iterative methods that approximate $e^x$ with the linear computations.
Our approximation comes from the limit definition of the natural exponential function $e^x = \lim_{n \rightarrow \infty} \left( 
1 + x/n \right)^n$.
We approximate the shares of $e^x$ from the shares of $x$ by $\left \langle e^x \right \rangle_b = \left( 
1 + \left \langle x \right \rangle_b / n \right)^n$.
Assume that $n=2^m$, we have
\begin{align}
    \label{equation-ex-expand}
    \left( 1 + \frac{x}{n} \right)^n = \left( 1 + \frac{x}{2^m} \right)^{2^m} = \left[ \left( 1 + \frac{x}{2^m} \right)^{2} \right]^{2^{m-1}}.
\end{align}
In other words, performing $m$ squares for $\left( 1 + x / 2^m \right)$ and its subsequent results yields the result of $ \left( 1 + x / n \right)^n$ where $n=2^m$.
For a share $\left \langle x \right \rangle_b \in \mathbb{Z}_{2^l}$ of $x$, the computation in Equation~\eqref{equation-ex-expand} involves a truncation to compute $\left( 1 + x / 2^m \right)$ and $m$ secure square functions.
We define the natural exponential function gate $\mathcal{G}_{\mathrm{exp}}$ as the family of functions $g_{\mathrm{exp}, l}: \mathbb{Z}_{2^l} \rightarrow \mathbb{Z}_{2^l}$ with input group $\mathbb{G}^{\mathrm{in}} = \mathbb{Z}_{2^l}$, output group $\mathbb{G}^{\mathrm{out}} = \mathbb{Z}_{2^l}$, and $g_{\mathrm{exp}, l} (x) := e^x$.
Using Equation~\eqref{equation-ex-expand} to approximate $e^x$, we denote the natural exponential function gate with FSS by $\hat{\mathcal{G}}_{\mathrm{exp}}$ and the offset functions by:
\begin{align}
    \hat{g}_{\mathrm{exp}, l}^{\left[ r^{\mathrm{in}}, r^{\mathrm{out}} \right]} (\hat{x})  & = g_{\mathrm{exp}, l} (\hat{x} - r^{\mathrm{in}}) + r^{\mathrm{out}} \mod 2^l \nonumber \\
    & = \left( 1 + \frac{\hat{x} - r^{\mathrm{in}}}{2^{m_{\exp}}} \right)^{2^{m_{\exp}}} + r^{\mathrm{out}} \mod 2^l,
\end{align}
where $m_{\exp}$ is the number of approximation iterations.
The detailed construction of the offline stage $\mathrm{Gen}^{\exp}_{l} (r^{\text{in}}, r^{\text{out}}, \mathbf{r})$ and the online stage $\mathrm{Eval}^{\mathrm{power}}_{l}(b, k_b, \left \langle \hat{x} \right \rangle_b)$ is given in APPENDIX~\ref{appendix-exp}.
The proposed protocol reduces the communication overhead of the secure $\exp$ in the online stage from $(l+5lm_{\exp})$ bits down to $(l+3lm_{\exp})$ bits compared to the implementation of CrypTen~\cite{Crypten}.

\subsubsection{Reciprocal Operation}
\label{section-reciprocal}
We employ the Newton-Raphson method to approximate the reciprocal of the input.
Given the function $f(y) = 1/y - x$, the iteration formula $y_{n+1} = y_{n} - f(y_n)/f'(y_n) = y_n \cdot (2 - xy_n)$ serves as a method to approximate $1/x$ as the number of iterations $n$ becomes sufficiently large.
We define the reciprocal gate $\mathcal{G}_{\mathrm{recip}}$ as the family of functions $g_{\mathrm{recip}, l}: \mathbb{Z}_{2^l} \rightarrow \mathbb{Z}_{2^l}$ with input group $\mathbb{G}^{\mathrm{in}} = \mathbb{Z}_{2^l}$, output group $\mathbb{G}^{\mathrm{out}} = \mathbb{Z}_{2^l}$, and $g_{\mathrm{recip}, l} (x) := 1/x$.
Using Newton's method to approximate $1/x$, we denote the reciprocal gate with FSS by $\hat{\mathcal{G}}_{\mathrm{recip}}$ and the offset functions by:
\begin{align}
    \hat{g}_{\mathrm{recip}, l}^{\left[ r^{\mathrm{in}}, r^{\mathrm{out}} \right]} (\hat{x})  & = g_{\mathrm{recip}, l} (\hat{x} - r^{\mathrm{in}}) + r^{\mathrm{out}} \mod 2^l \nonumber \\
    & = h_m + r^{\mathrm{out}} \mod 2^l
\end{align}
where
\begin{align}
    h_i = \left\{\begin{matrix}  y_0 \cdot \left[2 - (\hat{x}-r^{\mathrm{in}}) \cdot  x'_0\right] \mod 2^l, \; & \mathrm{if} \;  i=0;
 \\
h_{i-1} \cdot \left[2 - (\hat{x}-r^{(i)}) \cdot  h_{i-1}\right] \mod 2^l, \; & \mathrm{else},
\end{matrix}\right.
\end{align}
$y_0$ is the initial value, $m$ is the iteration round, and $r^{(i)}$ is the intermediate mask.
In alignment with CrypTen~\cite{Crypten}, we use the default initial value $y_0 = 3e^{1-2x} + 0.003$ when $y_0$ is unspecified.
Therefore, the communication overhead of a secure reciprocal operation depends on the provision of initial values, as the two scenarios vary by a single security $e^x$ invocation.
To minimize computational overhead when the range of $x$ is relatively small, providing an initial value is beneficial.
The detailed construction of the offline stage $\mathrm{Gen}^{\mathrm{recip}}_{l} (r^{\text{in}}, r^{\text{out}}, \mathbf{r})$ and the online stage $\mathrm{Eval}^{\mathrm{recip}}_{l}(b, k_b, \left \langle \hat{x} \right \rangle_b)$ is given in APPENDIX~\ref{appendix-reciprocal}.
The proposed protocol reduces the communication overhead of the secure reciprocal operation in the online stage from $(11l+5lm_{\exp}+10lm_{\mathrm{recip}})$ bits down to $(13l+3lm_{\exp}+6lm_{\mathrm{recip}})$ bits compared to the implementation of CrypTen~\cite{Crypten}.

\subsubsection{Softmax Function}
\label{section-softmax}
The softmax function takes a vector $\mathbf{x} = \left\{x_1, x_2, \cdots, x_K \right\}$ as input and normalizes $\mathbf{z}$ into a probability distribution consisting of $K$ probabilities, i.e., the softmax function rescales $\mathbf{x}$ so that the output lies in the range $\left[ 0,1 \right]$ and sums to $1$.
The softmax function is given by $\mathrm{Softmax}(\mathbf{x})_i = e^{x_i}/\sum_{j=1}^K e^{x_j}$.
SHAFT~\cite{SHAFT} introduces an iterative technique for approximating the softmax function, bypassing the need to compute the input's maximum value in CrypTen.
We decrease the cost of implementing the secure softmax function by improving the iterative approach utilized in SHAFT.
Specifically, given a vector $\mathbf{x}$, the initial value is set to $y_0=1/K$.
The protocol iteratively updates over $m$ iterations as follows:
\begin{equation}
    y_n = y_{n-1} + \frac{1}{m} \cdot (x- \left\langle x,y_{n-1} \right\rangle) \cdot y_{n-1},
\end{equation}
where $\left\langle x,y_{n-1} \right\rangle$ is an inner product calculation.
We denote the softmax gate with FSS by $\hat{\mathcal{G}}_{\mathrm{softmax}}$ and the offset functions by:
\begin{align}
    \hat{g}_{\mathrm{softmax}, l}^{\left[ r^{\mathrm{in}}, r^{\mathrm{out}} \right]} (\hat{x})  & = g_{\mathrm{softmax}, l} (\hat{x} - r^{\mathrm{in}}) + r^{\mathrm{out}} \mod 2^l \nonumber \\
    & = h_m + r^{\mathrm{out}} \mod 2^l
\end{align}
where
\begin{equation}
    h_i = h_{i-1} + \frac{1}{m} \cdot \left[\hat{x}-r^{(i)} - \left\langle (\hat{x}-r^{(i)}) \cdot  h_{i-1}\right \rangle \right] \mod 2^l.
\end{equation}
The detailed construction of the offline stage $\mathrm{Gen}^{\mathrm{softmax}}_{K \times l} (\mathbf{r}^{\text{in}},\mathbf{r}^{\text{out}}, \mathbf{r})$ and the online stage $\mathrm{Eval}^{\mathrm{softmax}}_{K \times l}(b, k_b, \left \langle \hat{\mathbf{x}} \right \rangle_b)$ is given in APPENDIX~\ref{appendix-exp}.
The proposed protocol reduces the communication overhead of the secure $\mathrm{Softmax}$  in the online stage from $(12l+6lm_{\mathrm{Softmax}})$ bits down to $(10l+10lm_{\mathrm{Softmax}})$ bits compared to the implementation of SHAFT~\cite{SHAFT}.

\subsubsection{Sigmoid Function and Hyperbolic Tangent}
\label{mechanism-softmax-tanh-reciprocal}
The sigmoid function is a common activation function, which is given by $\sigma(x)=1/{1+e^{-1}}$.
The hyperbolic tangent function is a popular activation function in neural networks, which is given by $\mathrm{Tanh}(x) = (e^{2x} - 1) /  (e^{2x} + 1)$.
The hyperbolic tangent, $\mathrm{Tanh}(x)$, can be derived using the sigmoid function, $\sigma(x)$, as follows: $\mathrm{Tanh}(x) = 2\sigma(2x) - 1$.
Therefore, we only discuss the implementation of $\sigma(x)$ in the paper, since the implementation of $\mathrm{Tanh}(x)$ is essentially the same as $\sigma(x)$.
Based on the definitions of $\sigma(x)$ and $\mathrm{Tanh}(x)$, these two computations can be securely executed by merging secure evaluations of secure $e^x$, reciprocal calculation, and multiplications.
We begin by determining $x$'s symbol $\mathrm{sign}(x)$ and converting $x$ to $|x|$ through the secure comparison with 0.
In later computations, the input is adjusted to $|x|$.
This adjustment can lower communication overhead by providing initial values and reduce the iteration times in the reciprocal calculation.
The final results can be obtained by
\begin{equation}
    \sigma(x)= \begin{cases}\sigma(|x|) & \text { if } \operatorname{sign}(x) \geq 0; \\ 1-\sigma(|x|) & \text { if } \operatorname{sign}(x)=-1.\end{cases}
\end{equation}
The proposed protocol reduces the communication overhead of the secure $\mathrm{Sigmoid}$ and $\mathrm{Tanh}$ in the online stage from $(25l+5lm_{\exp} + 10lm_{\mathrm{recip}})$ bits down to $(27l+3lm_{\exp} + 6lm_{\mathrm{recip}})$ bits compared to the implementation of CrypTen~\cite{Crypten}.
Due to the space limitation, the detailed construction of the above two protocols is given in Appendix~\ref{appendix-sigmoid-tanh}.

\subsubsection{Dropout Function}
In the training process, the dropout operation randomly zeroes some of the input elements with probability $p$.
Furthermore, the outputs are scaled by a factor of $\frac{1}{1-p}$.
For convenience, we define the dropout function by
\begin{equation}
    \mathrm{Dropout} (x;p) = 
    \left\{\begin{matrix}
        0, \; &\mathrm{with \; probability} \; p; \\
        \frac{x}{1-p}, \; &\mathrm{else}.
    \end{matrix}\right.
\end{equation}
We define the dropout gate $\mathcal{G}_{\mathrm{drop}}$ as the family of functions $g_{\mathrm{drop}, l}: \mathbb{Z}_{2^l} \rightarrow \mathbb{Z}_{2^l}$ with input $\mathbb{G}^{\mathrm{in}} = \mathbb{Z}_{2^l}$, output $\mathbb{G}^{\mathrm{out}} = \mathbb{Z}_{2^l}$, and $g_{\mathrm{drop}, l} (x;p) := \mathrm{Dropout}(x;p)$.
We denote the dropout gate with FSS by $\hat{\mathcal{G}}_{\mathrm{drop}}$ and the offset functions by:
\begin{align}
    \label{equation-drop-fss}
    \hat{g}_{\mathrm{drop}, l}^{\left[ r^{\mathrm{in}}, r^{\mathrm{out}} \right]} (\hat{x})  & = g_{\mathrm{drop}, l} (\hat{x} - r^{\mathrm{in}}) + r^{\mathrm{out}} \mod 2^l \nonumber \\
    & = \frac{\hat{x} - r^{\mathrm{in}}}{1 - p} \cdot \mathbf{1}\{r < p\} 
\end{align}
where $r \in (0,1)$ is a random number.
The straightforward way to implement a dropout gate is to generate random numbers belonging to the interval $(0,1)$ during the offline stage and then apply secure comparison to compare these random numbers with $p$ during the online stage.
However, the key of the dropout function is to calculate $\sigma = \mathbf{1}\{r < p\} / (1-p)$.
Obviously, the comparison can be done during the offline stage, and the results of $\sigma$ can be saved as keys to avoid the comparison operation during the online stage.

Based on the above analysis, we present the implementation of the dropout function with FSS in the APPENDIX~\ref{appendix-dropout}.
The proposed protocol reduces the communication overhead of the secure dropout in the online stage from $5ln$ bits down to $3ln$ bits compared to the implementation of CrypTen~\cite{Crypten}.
We refer to the above dropout method, which completes random number comparison in the offline stage, as the \textit{static dropout}.
A shortcoming is that the dropout probability must stay the same in secure federated fine-tuning.
Therefore, we also consider the \textit{dynamic dropout} in the evaluation, in which the key holds random values in the offline stage.
The secure comparison calculates the dropout factor according to the random values and dynamic probability in the online stage.

\subsubsection{Tensor Product of Vectors}
\label{section-tensor-prodcut}
We consider vectors $\mathbf{v} =\{v_1, v_2, \cdots, v_n\}$ and $\mathbf{w}=\{w_1, w_2, \cdots, w_m\}$  to discuss tensor product of vectors.
The tensor product $\mathbf{v} \otimes \mathbf{w}$ generates a $n \times m$ matrix such that
\begin{equation}
    \mathbf{v} \otimes \mathbf{w}=\left(\begin{array}{cccc}
v_1 w_1 & v_1 w_2 & \cdots & v_1 w_m \\
v_2 w_1 & v_2 w_2 & \cdots & v_2 w_m \\
\vdots & \vdots & \ddots & \vdots \\
v_n w_1 & v_n w_2 & \cdots & v_n w_m
\end{array}\right).
\end{equation}
The backpropagation of the model involves the tensor product of vectors, in which the linear layer gradient calculation involves the tensor product of the parameters and the gradients of the next layer.
A straightforward way to implement the tensor product of vectors is to treat $\mathbf{v}$ and $\mathbf{w}$ as $n \times 1$ and $1 \times m$ matrices and obtain $\mathbf{v} \otimes \mathbf{w}$ through matrix multiplication.
The key size of the above matrix multiplication is $3lmn$ bits, while the communication overhead is $\left( 2lm+2ln \right)$ bits.

Considering the training with batch size of $B$, the calculation of the gradient involves a $B\times n \times m$ tensor according to the tensor product of a $B \times n$ matrix and a $B \times m$ matrix.
The above computation can be implemented in two ways.
A simple method is to perform $B$ matrix multiplications through a loop.
However, this cannot load computing data to the GPU at one time, reducing computational efficiency.
Another method is to extend both matrices to a tensor of $B \times m \times n$ and calculate the result by secure multiplication.
The disadvantage is that the key size is $3Bmnl$ bits, and the communication overhead is $4Bmnl$ bits.
Meanwhile, the above method requires the GPU to load all relevant data simultaneously, increasing the GPU's video memory requirements.

The tensor multiplication of vectors by matrix multiplication is more suitable for cases where the batch size is 1.
However, it is hardly possible to set the batch size to 1 in reality.
We implement the tensor product using FSS to apply GPU acceleration while reducing the key size and communication overhead.
For ease of presentation, we still discuss the realization of the tensor product based on FSS through $\mathbf{v} \otimes \mathbf{w}$.
The key generation of the tensor product in the offline phase is similar to that of the secure multiplication.
The detailed implementations are given in APPENDIX~\ref{appendix-tensor-product}.
The innovation is that the implementation in PriFFT applies $\mathbf{r}^{\text{in}_1} \otimes \mathbf{r}^{\text{in}_2} \mod 2^l$ to generate an $N \times M$ matrix.
In the online stage, the implementation in PriFFT calculates the result through the tensor product of the inputs with offsets and shares of masks.
The detailed construction is given in APPENDIX~\ref{appendix-tensor-product}.
The proposed protocol reduces the communication overhead of the secure tensor production in the online stage from $5lMN$ bits down to $(2lM+2lN+lMN)$ bits compared to the implementation of CrypTen~\cite{Crypten}.

\section{Theoretical Analysis}
\label{section-analysis}

\subsection{Communication Overhead}
The communication overhead of the proposed protocols is summarized in TABLE~\ref{table-analysis-communication-overhead}.
We compare the implementation of PriFFT, ABY2~\cite{patra2021aby2}, CrypTen~\cite{Crypten}, and SHAFT~\cite{SHAFT}.

\begin{table*}
	\caption{Communication overhead analysis of the proposed protocols in the online stage.}
        \label{table-analysis-communication-overhead}
	\centering
	\begin{tblr}{hline{3-Y}={dotted}, vline{2-3}={2-Z}{}, vline{5}={2}{}, column{3-5}={1in, c,m}}
        \toprule
        Protocols & PriFFT & ABY2 & CrypTen & SHAFT \\
        \midrule
        $\mathrm{Eval}^{\mathrm{softmax}}_{K \times l}(b, k_b, \left \langle \hat{\mathbf{x}} \right \rangle_b)$ & $ 12l+ 6lm_{\mathrm{Softmax}} $ & \SetCell[c=2]{c,m} $16l + \log_2(K)\cdot9l + 10lm_{\mathrm{exp}} + 10lm_{\mathrm{recip}}$ & & $10l+ 10lm_{\mathrm{Softmax}}$ \\
        $\mathrm{Eval}^{\mathrm{sigmoid}}_{l}(b, k_b, \left \langle \hat{\mathbf{x}} \right \rangle_b)$ & \SetCell[r=2]{c,m} $ 27l+ 3lm_{\mathrm{\exp}}+6lm_{\mathrm{recip}} $ & \SetCell[c=3,r=2]{c,m} $25l + 5lm_{\mathrm{exp}} + 10lm_{\mathrm{recip}}$ & &  \\
        $\mathrm{Eval}^{\mathrm{tanh}}_{l}(b, k_b, \left \langle \hat{\mathbf{x}} \right \rangle_b)$ & & & & \\
        $\mathrm{Eval}^{\mathrm{exp}}_{l}(b, k_b, \left \langle \hat{x} \right \rangle_b)$ & $l + 3l m_{\exp}$ & \SetCell[c=3]{c,m}  $l+5lm_{\exp}$ & & \\
        $\mathrm{Eval}^{\mathrm{recip}}_{l}(b, k_b, \left \langle \hat{x} \right \rangle_b, y_0)$ & $12l + 6lm_{\mathrm{recip}}$ & \SetCell[c=3]{c,m} $10l \cdot (1+m_{\mathrm{recip}})$ &  &  \\
        $\mathrm{Eval}^{\mathrm{recip}}_{l}(b, k_b, \left \langle \hat{x} \right \rangle_b)$ & $13l +3lm_{\exp} + 6lm_{\mathrm{recip}}$ & \SetCell[c=3]{c,m} $11l + 5lm_{\exp} + 10lm_{\mathrm{recip}})$ &  & \\
        $\mathrm{Eval}^{\mathrm{TP}}_{N \times  l,M \times l} (\mathbf{r}^{\text{in}_1},\mathbf{r}^{\mathrm{in}_2}, \mathbf{r}^{\text{out}})$ & $2l(M+N) + MNl$ & \SetCell[c=3]{c,m} $ 5lMN$ & &\\
        $\mathrm{Eval}^{\mathrm{power}}_{l}(b, k_b, \left \langle \hat{x} \right \rangle_b, 2^n)$ & $3ln$ & \SetCell[c=3]{c,m} $5ln$ \\
        $\mathrm{Eval}^{\mathrm{drop}}_{l}(b, k_b, \left \langle \hat{x} \right \rangle_b)$ (static) & $3l$ & \SetCell[c=3]{c,m} $5l$ & & \\
        $\mathrm{Eval}^{\mathrm{drop}}_{l}(b, k_b, \left \langle \hat{x} \right \rangle_b)$ (dynamic) & $5l$ & \SetCell[c=3]{c,m} $9l$ & & \\
        \bottomrule
	\end{tblr}
\end{table*}

\subsection{Security Analysis}
We prove the security of all proposed secure protocols against semi-honest attackers under the cryptographic standard of security definition~\cite{lindell2017simulate}. Please refer to Appendix~\ref{appendix-security-analysis} for detailed security proofs.

\section{Evaluation}
\label{section-evaluation}
\subsection{Evaluation Setup}
\label{section-evaluation-setup}
\textbf{Benchmark and datasets.} We evaluate the LLMs fine-tuned by PriFFT on general language understanding evaluation (GLUE)~\cite{wang2018glue}.
We consider four classic GLUE training tasks in the evaluation: Stanford sentiment treebank (SST-2), Microsoft research paraphrase corpus (MRPC), recognizing textual entailment (RTE), and corpus of linguistic acceptability (CoLA). 
In evaluation, we compare various secure multi-party computation mechanisms ABY2~\cite{ABY}, CrypTen~\cite{Crypten} and SHAFT~\cite{SHAFT}.
ABY2 is an open-source library for secure multi-party computation based on arithmetic secret sharing.
Similarly, CrypTen is a secure multi-party computing library developed by Microsoft, which also supports neural network inference and training.
SHAFT optimizes some protocols based on crypten and realizes the secure inference of large models.
The comparison mechanisms involved above includes the mainstream secure multi-party computation library (ABY2), the secure multi-party computation library designed for neural networks (CrypTen), and the state-of-the-art work in secure inference of large models (SHAFT).



\textbf{Training models.}
We consider various language models in the evaluation: BERT (base and large) with 110M and 340M parameters~\cite{devlin2019bertpretrainingdeepbidirectional}; RoBERTa (base and large) with 125M and 355M parameters~\cite{roberta2019Liu}; DistilBERT with 67M parameters~\cite{DistilBERT}; ALBERT (base and large) with 12M and 18M parameters~\cite{ALBERT}; DeBERTa V1 (base and large) with 100M and 350M parameters~\cite{DeBERTa}.
We use the original public versions of these models, meaning none have been fine-tuned for specific downstream tasks.
The backbone of the models is frozen, and we fine-tune the poolers and classifiers according to the training tasks.

\textbf{Metrics.} 
For SST-2, MRPC, and RTE, we consider inference accuracy (Acc) as the metric. 
The metric for CoLA is the Matthews correlation coefficient (MCC). 
MCC is a metric to measure the performance of classification models, especially in cases of category imbalance, providing a value between -1 and 1, where 1 indicates perfect classification, 0 indicates random classification, and -1 indicates complete misclassification. 
The accuracy and MCC are calculated using the evaluation library of Hugging Face~\cite{HuggingFace}.
Besides, we also focus on the communication overhead and execution time of fine-tuning LLMs and each protocol proposed in this paper.
Communication overhead includes the traffic of sending data and receiving data, while execution time includes local execution time and data transmission time.

\textbf{Experiment environment.}
All experiments are conducted on a system equipped with an NVIDIA RTX 4090 GPU with 24 GB memory and an Intel Xeon Gold 6430 CPU.
We evaluate our privacy-preserving mechanisms and protocols in the LAN setting with 0.21 ms round-trip latency and 2.5 Gbps network bandwidth.
PriFFT is implemented by PyTorch and NssMPClib, while all computations are accelerated by CUDA.
Evaluation works in the secret sharing domain, where all data are represented in integers and shared over the ring $\mathbb{Z}_{2^{64}}$.
The scale factor $s$ is set to $16$.

\subsection{Protocol Analysis}
This section analyzes the execution time and communication overhead of protocols mentioned in the paper under various mechanisms.
In SHAFT and CrypTen, the standard network setting is configured to utilize local computing resources.
To maintain experimental consistency, the execution time reported in this section is measured within the identical local computing environment.
More specifically, we emulate multiple parties using multiple processes on an NVIDIA RTX 4090, facilitating data exchange via shared memory for communication.
Furthermore, we evaluate the execution times of various protocols within a LAN setting. 

\begin{table}
	\caption{Overhead of secure $\mathrm{Softmax}(x)$.}
        \label{table-evaluation-softmax-overhead}
	\centering
	\begin{tblr}{hline{3,4,5,7,8,9}={dotted}, vline{2}={2-Z}{}, vline{3}={2-Z}{}, columns={c,m}, column{1}={0.3in}, column{3-5} = {0.28in}, column{6-Z}={0.32in}
		}
            \toprule
		  \SetCell[c=2]{c,m} Input Size & & $10^2$ & $10^3$ & $10^4$ & $10^5$ & $10^6$ \\
            \midrule
		\SetCell[r=4]{c,m} Time (ms) & \textbf{PriFFT} & \textbf{105.4} & \textbf{107.1} & \textbf{107.8} & \textbf{137.0} & \textbf{472.8} \\ 
            & ABY2 & 334.0 & 382.4 & 435.4 & 515.0 & 910.6 \\
            & CrypTen & 337.1 & 376.5 & 452.1 & 501.0 & 912.3 \\
            & SHAFT & 131.7 & 133.9 & 134.3 & 171.8 & 637.6 \\
            \midrule
            \SetCell[r=4]{c,m} Comm (MB) &  PriFFT & \textbf{0.08} & \textbf{0.82} & \textbf{8.24} & \textbf{82.40} & \textbf{823.97} \\
            & ABY2 & 0.19 & 2.12 & 23.96 & 260.66 & 2807.62 \\
            & CrypTen & 0.19 & 2.12 & 23.96 & 260.66 & 2807.62 \\
            & SHAFT & 0.13 & 1.30 & 12.97 & 129.70 & 1297.00 \\
            \bottomrule
	\end{tblr}
\end{table}

TABLE~\ref{table-evaluation-softmax-overhead} displays the overhead associated with various implementations of the secure softmax function.
We optimize the implementation of the secure softmax function within PriFFT, enabling PriFFT to realize the lowest execution time and lowest communication overhead for the secure softmax function compared to other mechanisms.
The clipping range is set to $(-4,12)$, and $m_{\mathrm{Softmax}} = 16$, matching the evaluation configuration in SHAFT.
SHAFT provides an approximation of the softmax function, bypassing the need to determine the maximum value in the input vector. 
This method reduces communication overhead and decreases execution time when compared to both ABY2 and CrypTen.
Since ABY2 does not optimize the softmax function, the softmax function in ABY2 relies on the foundational operations of ABY2 combined with the implementation strategy of the softmax in CrypTen.
Consequently, both ABY2 and CrypTen incur identical communication costs when applying the softmax function.
Compared to SHAFT, PriFFT decreases execution time by approximately $25\%$ and communication overhead by around $36\%$ in the secure $e^x$ protocol.


\begin{table*}
	\caption{Execution time and communication overhead of secure $1/x$.}
	\centering
	\label{table-evaluation-reciprocal-overhead}
	\begin{tblr}{hline{4,5,6,8,9,10}={dotted}, vline{2}={2-Z}{}, vline{3}={2-Z}{}, vline{8}={2-Z}{}, columns={c,m}, column{1}={0.4in}, column{2}={0.5in}, column{3-Z} = {0.37in}
		}
            \toprule
            \SetCell[c=2]{c,m} & & \SetCell[c=5]{c,m} Without initial values & & & & & \SetCell[c=5]{c,m} With initial values & & & & \\
            \midrule
		  \SetCell[c=2]{c,m} Input Size & & $10^3$ & $10^4$ & $10^5$ & $10^6$ & $10^7$ & $10^3$ & $10^4$ & $10^5$ & $10^6$ & $10^7$ \\
            \midrule
		\SetCell[r=4]{c,m} Time (ms) & \textbf{PriFFT} & \textbf{74.0} & \textbf{74.2} & \textbf{80.1} & \textbf{200.4} & \textbf{1876.5} & \textbf{54.1} & \textbf{54.3} & \textbf{59.0} & \textbf{132.8} & \textbf{1207.9}\\
            & ABY2 & 80.2 & 82.7 & 94.0 & 232.5 & 2440.1 & 65.8 & 72.9 & 96.4 & 197.6 & 1683.9\\
            & CrypTen & 110.5 & 110.7 & 130.1 & 429.3 & 4017.7 & 96.8 & 102.7 & 114.8 & 356.7 & 3222.9\\
            & SHAFT &  100.5 & 109.1 & 140.3 & 433.8 & 4023.1 & 93.1 & 101.5 & 109.8 & 358.5 & 3276.0\\
            \midrule
            \SetCell[r=4]{c,m} Comm (MB) & \textbf{PriFFT} & \textbf{0.74} & \textbf{7.40} & \textbf{74.01}  & \textbf{740.05} & \textbf{7400.51} & \textbf{0.55} & \textbf{5.49} & \textbf{54.93} & \textbf{549.32} & \textbf{5493.16}\\
            & ABY2 & 1.15 & 11.52 & 115.20 & 1152.04 & 11520.4 & 0.83 & 8.39 & 83.92 & 839.23 & 8392.33\\
            & CrypTen & 1.15 & 11.52 & 115.20 & 1152.04 & 11520.4 & 0.83 & 8.39 & 83.92 & 839.23 & 8392.33\\
            & SHAFT & 1.15 & 11.52 & 115.20 & 1152.04 & 11520.4 & 0.83 & 8.39 & 83.92 & 839.23 & 8392.33\\
            \bottomrule
	\end{tblr}
\end{table*}

Table~\ref{table-evaluation-reciprocal-overhead} illustrates the overhead for the secure $1/x$ across different mechanisms.
Similar to the secure $\mathrm{Softmax}(x)$, the methodology of $1/x$ in ABY2 is the same as the logic employed by CrypTen, with the exception that $1/x$ is reconstructed using ABY2's foundational protocol.
Therefore, ABY2, CrypTen, and SHAFT have the same communication consumption, and have a minor difference in execution time.
PriFFT enhances secure $1/x$ operations by optimizing input mask size in the secure multiplication, thereby decreasing communication overhead and eliminating some steps of shared data restoration to achieve reduced execution time.

As discussed in Section \ref{section-reciprocal}, the accuracy of the reciprocal operation is related to initial values.
In the absence of a specified initial value, the secure protocol computes $\left( 3 e^{1-2 x} + 0.003 \right)$ as the default, leading to extra overhead.
Providing approximate initial values when the input range is known can not only eliminate calculating initial values, but also decrease the number of iterations, thereby minimizing computational cost.
Table~\ref{table-evaluation-reciprocal-overhead} presents the overhead of the secure reciprocal operation, both with and without initial values.
The difference is nearly equivalent to a call to the secure $e^x$.
Compared to SHAFT, PriFFT decreases execution time by approximately $63\%$ and communication overhead by around $35\%$ in the secure $1/x$ protocol.

\begin{table}
	\caption{Overhead of secure $e^x$.}
	\centering
	\label{table-evaluation-exp-overhead}
	\begin{tblr}{hline{3,4,5,7,8,9}={dotted}, vline{2}={2-Z}{}, vline{3}={2-Z}{}, columns={c,m}, column{1}={0.3in}, column{3-5} = {0.28in}, column{6-Z}={0.32in}
		}
            \toprule
		  \SetCell[c=2]{c,m} Input Size & & $10^3$ & $10^4$ & $10^5$ & $10^6$ & $10^7$ \\
            \midrule
		\SetCell[r=4]{c,m} Time (ms) & \textbf{PriFFT} & \textbf{16.0} & \textbf{16.2} & \textbf{18.2} & \textbf{46.8} & \textbf{478.7}\\
            & ABY2 & 33.3 & 35.3 & 37.5 & 71.4 & 543.7 \\
            & CrypTen & 32.9 & 35.2 & 36.5 & 78.2 & 542.5 \\
            & SHAFT & 34.1 & 37.0 & 38.3 & 75.3 & 547.8 \\
            \midrule
            \SetCell[r=4]{c,m} Comm (MB) & \textbf{PriFFT} & \textbf{0.19} & \textbf{1.91} & \textbf{19.07}  & \textbf{190.73} & \textbf{1907.35} \\
            & ABY2 & 0.31 & 3.13 & 31.28 & 312.81 & 3128.05 \\
            & CrypTen & 0.31 & 3.13 & 31.28 & 312.81 & 3128.05 \\
            & SHAFT & 0.31 & 3.13 & 31.28 & 312.81 & 3128.05 \\
            \bottomrule
	\end{tblr}
\end{table}

TABLE~\ref{table-evaluation-exp-overhead} and TABLE~\ref{table-evaluation-sigmoid-overhead} illustrate the overhead of the secure $e^x$ and $\mathrm{Sigmoid}(x)$ protocols.
In the isolated assessment of the secure $e^x$ protocol's overhead, $m_{\exp}$ is configured to 8, aligning with CrypTen's default parameters.
Compared to SHAFT, PriFFT decreases execution time by approximately $13\%$ and communication overhead by around $39\%$ in the secure $1/x$ protocol.
Compared to SHAFT, PriFFT reduces $53\%$ execution time and $20\%$ communication overhead in the secure $e^x$ protocol.
The secure sigmoid protocol exhibits a moderate level of complexity, given that it requires secure zero comparisons, $1/x$, and $e^x$.
We maintain the sigmoid parameters identical to the default configuration in CrypTen.
Specifically, $m_{\mathrm{recip}}$ in the secure reciprocal is set to 3 with an initial value of $0.75$.
The sigmoid function produces outputs in the range $(0,1)$, allowing the protocol to decrease the number of iterations and thereby minimize computational overhead.
Compared to SHAFT, PriFFT decreases execution time by approximately $8\%$ and communication overhead by around $33\%$ in the secure $1/x$ protocol.


\begin{table}
	\caption{Overhead of secure $\mathrm{Sigmoid}(x)$.}
	\centering
	\label{table-evaluation-sigmoid-overhead}
	\begin{tblr}{hline{3,4,5,7,8,9}={dotted}, vline{2}={2-Z}{}, vline{3}={2-Z}{}, columns={c,m}, column{1}={0.3in}, column{3-5} = {0.28in}, column{6-Z}={0.32in}
		}
            \toprule
		  \SetCell[c=2]{c,m} Input Size & & $10^3$ & $10^4$ & $10^5$ & $10^6$ & $10^7$ \\
            \midrule
		\SetCell[r=4]{c,m} Time (ms) & \textbf{PriFFT} & \textbf{63.5} & \textbf{68.0} & \textbf{77.9} & \textbf{292.5} & \textbf{2677.0}\\
            & ABY2 & 93.3 & 94.5 & 110.7 & 342.8 & 3016.3 \\
            & CrypTen & 87.9 & 94.0 & 106.3 & 314.6 & 2795.9 \\
            & SHAFT & 87.6 & 92.5 & 104.4 & 317.41 & 2736.9 \\
            \midrule
            \SetCell[r=4]{c,m} Comm (MB) &  PriFFT & \textbf{0.55} & \textbf{5.49} & \textbf{54.93} & \textbf{549.32} & \textbf{5493.16} \\
            & ABY2 & 0.76 & 7.63 & 76.29 & 762.94 & 7629.39 \\
            & CrypTen & 0.76 & 7.63 & 76.29 & 762.94 & 7629.39 \\
            & SHAFT & 0.76 & 7.63 & 76.29 & 762.94 & 7629.39 \\
            \bottomrule
	\end{tblr}
\end{table}

\subsection{Performance and Overhead of Fine-tuning LLMs}
\label{section-evaluation-model-fine-tuning}
We evaluate the performance of fine-tuned models with different settings, then discuss the training process's overhead in the LAN setting.
PriFFT is the first mechanism that presents privacy-preserving federated fine-tuning to protect model parameters and updates based on SMPC.
Therefore, we implement an ABY2-based mechanism to realize the fine-tuning of LLMs for comparison.
TABLE~\ref{table-evaluation-model-performance-comparison} summarizes the best performance of fine-tuned models under various tasks, where PT, PriFFT-IT, and PriFFT-LT refer to the training implemented in plaintext, iterative truncations, and local truncations, respectively.
PriFFT-IT applies iterative truncations for all multiplications.
Matrix multiplication in PriFFT-LT is implemented via iterative truncations to improve the accuracy, and protocols proposed in Section~\ref{section-mechanism} are implemented via local truncations to reduce the overhead.


\begin{table*}
	\caption{Fine-tuned models performance comparison under various datasets and settings. PriFFT with the iterative and local truncation (PriFFT-IT and PriFFT-LT) cause slight performance loss compared to the fine-tuning with plaintext (PT).}
	\centering
	\label{table-evaluation-model-performance-comparison}
	\begin{tblr}{
			columns = {0.43in, c, m}, hline{2}, hline{2}={4,7,10}{rightpos=-1}, column{3,6,9,12}={0.49in}, column{4,7,10,13}={0.5in},
			hline{4-Y}={dotted},
			column{1} = {0.75in},
			colsep = {2pt},
		}
        \toprule
		\SetCell[r=2]{c,m} Tasks & \SetCell[c=3]{c} SST-2 (Acc) & & & \SetCell[c=3]{c} MRPC (Acc) & & & \SetCell[c=3]{c} RTE (Acc) & & & \SetCell[c=3]{c} CoLA (MCC) \\
		& PT & PriFFT-IT & PriFFT-LT & PT & PriFFT-IT & PriFFT-LT & PT & PriFFT-IT & PriFFT-LT & PT & PriFFT-IT & PriFFT-LT \\
        \midrule
		BERT-base & $85.13\%$ & $84.21\%$ & $83.72\%$ & $73.77\%$ & $72.16\%$ & $70.58\%$ & $61.01\%$ & $58.85\%$ & $57.43\%$ & 0.441 & 0.428 &  0.420 \\
            BERT-large & $86.76\%$ & $84.79\%$ & $83.23\%$ & $75.74\%$ & $73.23\%$ & $72.92\%$ & $62.09\%$ & $60.93\%$ & $60.11\%$ & 0.461 & 0.449 & 0.432 \\
		RoBERTa-base & $83.71\%$ & $81.93\%$ & $81.02\%$ & $70.83\%$ & $69.01\%$ & $67.91\%$ & $54.87\%$ & $53.01\%$ & $51.94\%$ & 0.386 & 0.370 & 0.357\\
		RoBERTa-large & $85.55\%$ & $84.01\%$ & $82.96\%$ & $71.57\%$ & $69.86\%$ & $68.51\%$ & $57.76\%$ & $56.32\%$ & $54.79\%$ & 0.404 & 0.389 & 0.372\\
		DistilBERT & $84.28\%$ & $82.45\%$ & $81.79\%$ & $72.06\%$ & $70.48\%$ & $69.98\%$ & $58.84\%$ & $57.15\%$ & $55.64\%$ & 0.358 & 0.342 & 0.331 \\
		ALBERT-base & $83.83\%$ & $82.35\%$ & $81.78\%$ & $76.96\%$ & $76.13\%$ & $74.85\%$ & $67.51\%$ & $66.23\%$ & $65.11\%$ & 0.380 & 0.369 & 0.361\\
		ALBERT-large & $85.24\%$ & $83.89\%$ & $82.51\%$ & $77.70\%$ & $76.84\%$ & $76.01\%$ & $68.73\%$ & $67.91\%$  & $66.86\%$ & 0.441 & 0.435 & 0.429\\
		DeBERTa-base & $82.90\%$ & $82.11\%$ & $80.83\%$ & $70.30\%$ & $68.17\%$ & $66.95\%$ & $60.65\%$ & $58.73\%$ & $58.02\%$ & 0.499 & 0.487 & 0.481 \\
        DeBERTa-large & $85.56\%$ & $84.78\%$ & $83.62\%$ & $71.08\%$ & $69.53\%$ & $68.47\%$ & $63.34\%$ & $61.98\%$ & $61.26\%$ & 0.520 & 0.505 & 0.482\\
        \bottomrule
	\end{tblr}
\end{table*}

\begin{figure}
    \centering
    \includegraphics[width=0.95\linewidth]{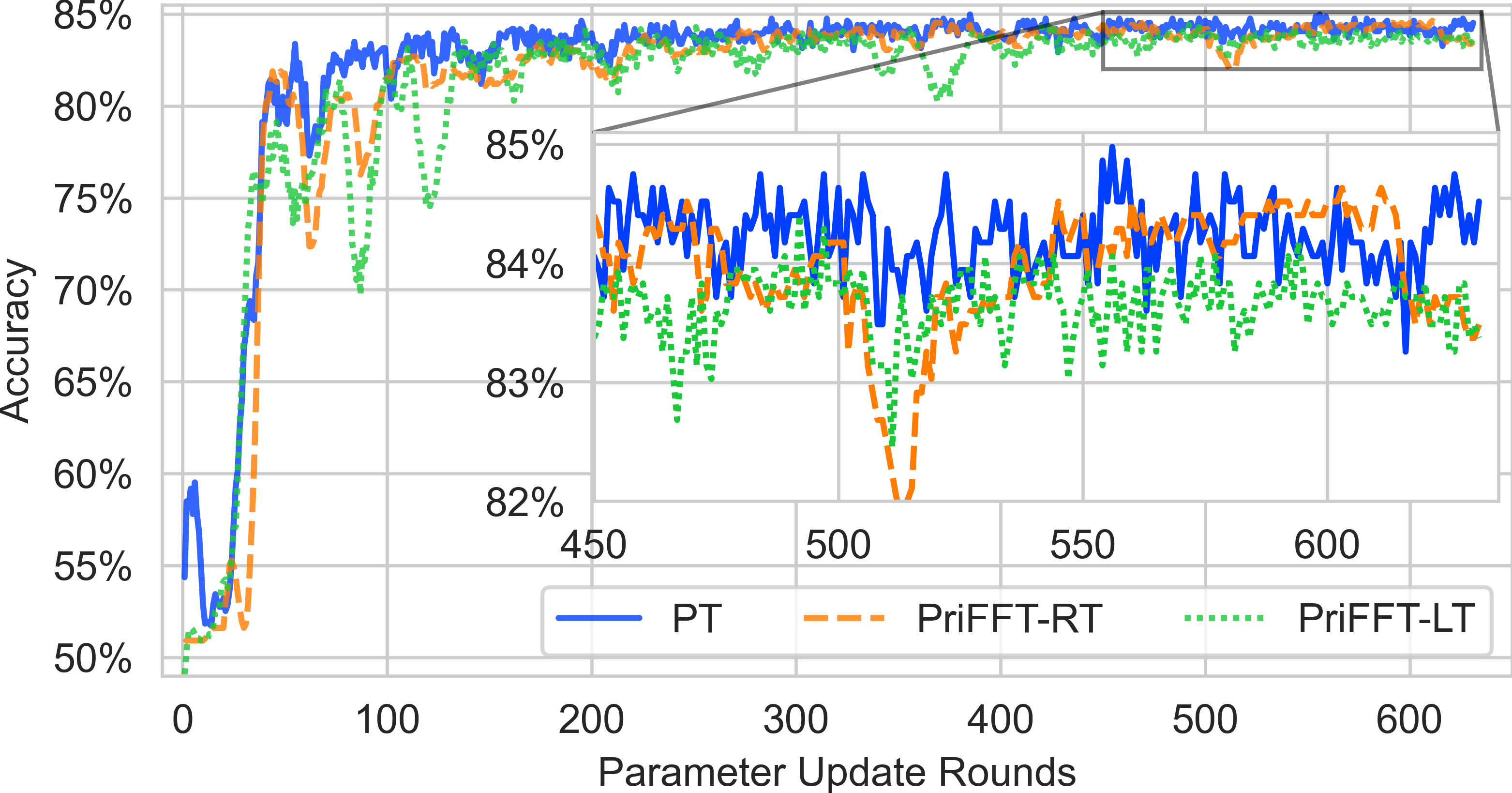}
    \caption{The relationship between the accuracy of BERT on SST-2 and parameter update rounds.}
    \label{fig-evaluation-convergence-comparision}
\end{figure}

Fig.~\ref{fig-evaluation-convergence-comparision} illustrates how the accuracy on SST-2 of the fine-tuned model (BERT-base) varies with the aggregation times.
There is little difference between ciphertext and plaintext fine-tuning before the model converges, and models fine-tuned with plaintext can achieve higher accuracy.
The calculation error caused by the fine-tuning on encrypted samples does not bring significant model performance degradation for three reasons.
Firstly, extensive training on the tunable parts of the model parameters can offset the impact of computational errors to a certain extent during the fine-tuning.
Secondly, gradient values determine the change in model parameters, and gradients in the optimization step are generated by a batch of training samples.
Each parameter update of the model is the average of multiple calculation results, reducing the impact of a single calculation error.
Finally, model training often introduces randomness into the calculations to prevent overfitting.
The success of introducing randomness in model training shows that model training is not sensitive to small amounts of noise.

\begin{table}
	\caption{Communication consumption of privacy-preserving fine-tuning for BERT with a batch of samples.}
	\centering
	\label{table-evaluation-one-batch-communication}
	\begin{tblr}{hline{3,4,6,7}={dotted}, vline{2}={2-Z}{}, vline{3}={2-Z}{}, column{1}={0.65cm}, column{2}={1.3cm}, column{3-6}={1.05cm}, columns={c,m}
		}
            \toprule
		  \SetCell[c=2]{c,m} Batch Size &     & 8        & 16       & 32       & 64 \\
            \midrule
		\SetCell[r=3]{c,m} BERT-base  & PriFFT-IT   & 61.43MB  & 104.78MB & 191.86MB & 364.90MB \\
                                          & PriFFT-LT   & 25.27MB  & 28.46MB  & 39.21MB  & 59.60MB \\
                                          & ABY2 & 210.05MB & 402.06MB & 786.06MB & 1.46GB \\
            \midrule
            \SetCell[r=3]{c,m} BERT-large & PriFFT-IT   & 105.90MB & 179.70MB & 327.29MB & 622.48MB \\
                                          & PriFFT-LT   & 39.02MB  & 44.94MB  & 59.77MB  & 87.44MB \\
                                          & ABY2 & 368.01MB & 701.06MB & 1.34GB   & 2.66GB \\
        \bottomrule
	\end{tblr}
\end{table}

\begin{table}
	\caption{Execution time of privacy-preserving fine-tuning for BERT with a batch of samples.}
	\centering
	\label{table-evaluation-one-batch-time}
	\begin{tblr}{hline{5,11}, hline{3,4,6,7,9,10,12,13}={dotted}, vline{2}={2-Z}{}, vline{3}={2-Z}{}, column{1}={0.55cm}, column{2}={0.66cm}, column{3}={1.3cm}, column{4-7}={0.81cm}, columns={c,m}
		}
            \toprule
		  \SetCell[c=3]{c,m} Batch Size & & & 8 & 16 & 32 & 64 \\
            \midrule
		\SetCell[r=6]{c,m} CPU & \SetCell[r=3]{c,m} BERT-base  & PriFFT-IT  & 0.83s & 1.34s & 2.42s & 4.64s \\
                                   &                               & PriFFT-LT   & 0.67s & 0.72s & 1.17s & 1.75s \\
                                   &                               & ABY2 & 1.04s & 1.85s & 4.28s & 7.60s \\
                                   & \SetCell[r=3]{c,m} BERT-large & PriFFT-IT  & 1.52s & 2.81s & 4.53s & 9.07s \\
                                   &                               & PriFFT-LT  & 0.74s & 1.13s & 1.95s & 3.78s \\
                                   &                               & ABY2 & 1.53s & 3.39s & 6.80s & 11.52s \\
            \midrule
            \SetCell[r=6]{c,m} GPU & \SetCell[r=3]{c,m} BERT-base  & PriFFT-IT  & 0.78s & 1.01s & 1.41s & 2.98s \\
                                   &                               & PriFFT-LT  & 0.51s & 0.55s & 0.58s & 1.23s \\
                                   &                               & ABY2 & 1.12s & 1.41s & 3.15s & 6.73s \\
                                   & \SetCell[r=3]{c,m} BERT-large & PriFFT-IT  & 1.05s & 1.36s & 1.88s & 3.12s \\
                                   &                               & PriFFT-LT  & 0.64s & 0.68s & 0.73s & 0.99s \\
                                   &                               & ABY2 & 1.48s & 1.86s & 5.29s & 10.51s \\
            \bottomrule
	\end{tblr}
\end{table}


We next discuss the communication and computational consumption during privacy-preserving fine-tuning in the LAN setting.
TABLE~\ref{table-evaluation-one-batch-communication} summarizes the communication overhead when fine-tuning BERT with a batch of samples under different privacy-preserving implementations.
TABLE~\ref{table-evaluation-one-batch-time} summarizes the execution time of corresponding privacy-preserving fine-tuning.
We also compare the communication consumption and execution time of privacy-preserving fine-tuning of BERT through all training samples of downstream tasks.
Due to the space limitation, the results are given in APPENDIX~\ref{appendix-evaluation}.

\subsection{Error Evaluation}
\label{evaluation-error-measure}

PriFFT approximates $e^x$, $1/x$, and $\mathrm{Tanh}(x)$ by the secure computation.
Meanwhile, we provide different truncation methods in the paper.
In this section, we combine evaluation to discuss the errors of each secure calculation under different settings.
Evaluation first calculates the relative error between the approximated results from the secure computation and the exact results computed in plaintext with PyTorch.
We use the error rate to analyze the error in the secure computation, i.e., the percentage between the relative errors and the exact values.
The inputs are $10^5$ random values from 0 to 1, and the provided results are the mean of $10^5$ error rates for each case.

\begin{figure}
    \centering
    \begin{subfigure}[b]{0.488\linewidth}
        \centering
        \includegraphics[width=\linewidth]{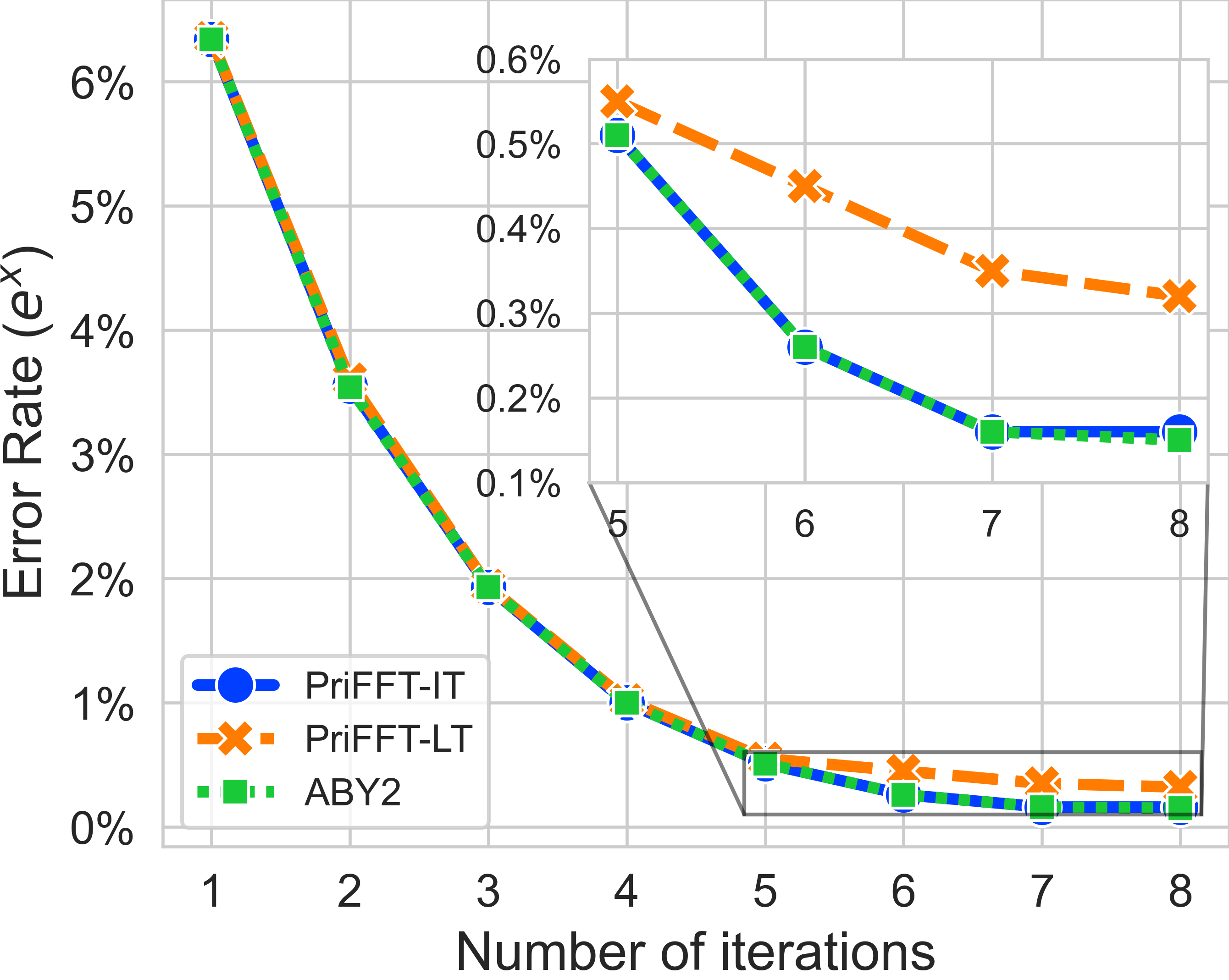}
        \caption{$e^x$}
        \label{fig-evaluation-exp}
    \end{subfigure}%
    \hspace{0.01\linewidth}
    \begin{subfigure}[b]{0.488\linewidth}
        \centering
        \includegraphics[width=\linewidth]{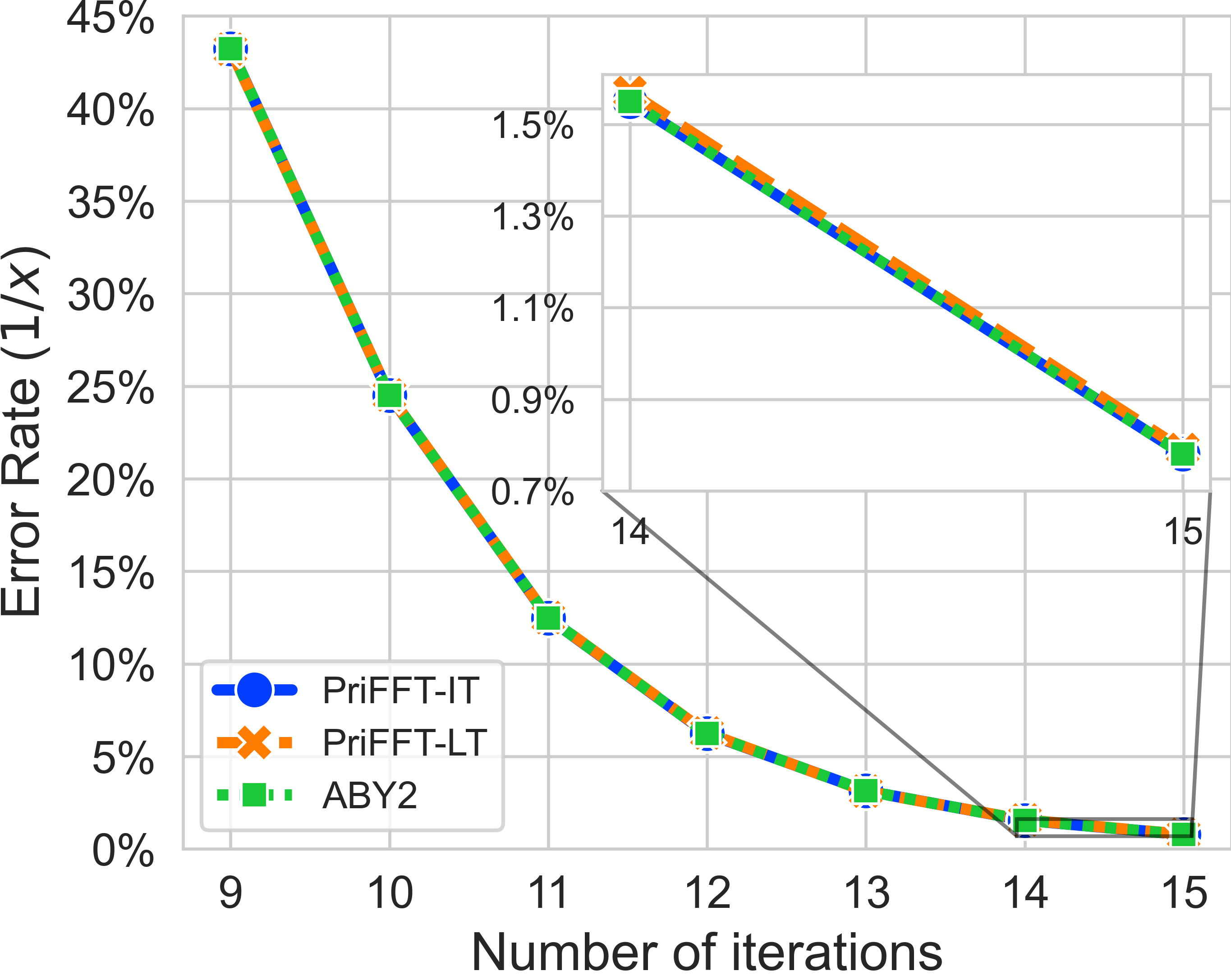}
        \caption{$1/x$}
        \label{fig-evaluation-recip}
    \end{subfigure}
    \caption{The impact of the number of iterations on the error rates of the secure $e^x$ and reciprocal calculation.}
    \label{fig-evaluation-exp-and-recip}
\end{figure}

Fig.~\ref{fig-evaluation-exp-and-recip} presents the impact of the number of iterations $m_{\exp}$ and $m_{\mathrm{recip}}$ on the error rates of secure $e^x$ and $\mathrm{Tanh}(x)$.
For secure $e^x$, when $m_{\exp} < 5$, truncation methods have no significant impact on the error rate, as shown in Fig.~\ref{fig-evaluation-exp}.
When the number of iterations is insufficient, the error is relatively large. 
In this case, the difference between iterative and local truncations is too small to impact the error rate significantly.
The error rate of secure $e^x$ decreases as $m_{\exp}$ increases.
Meanwhile, differences in truncation methods begin to have an impact on error.
Experimental results show that it is difficult to effectively reduce the error rates by further increasing the number of iterations.
It comes from that the experiment takes 16 bits to represent the fractional part of fixed-point numbers.
More accurate results rely on more bits to represent the fractional part.

Fig.~\ref{fig-evaluation-recip} presents the error rate of the secure reciprocal calculation under different $m_{\mathrm{recip}}$.
Despite applying a higher number of iterations, the secure reciprocal calculation does not have a lower error rate than the secure $e^x$.
The reason is that different approximation calculations were performed during the iteration process.
As shown in Equation~\ref{equation-ex-expand}, setting the number of iterations to 8 for the secure $e^x$ implies that $n=2^8=256$ for $(1+x/n)^n$.
However, the number of iterations equals the number of approximations for the secure reciprocal calculation, leading the secure $e^x$ to realize lower error rates with fewer iterations.

\begin{figure}
    \centering
    \begin{subfigure}[b]{0.488\linewidth}
        \centering
        \includegraphics[width=\linewidth]{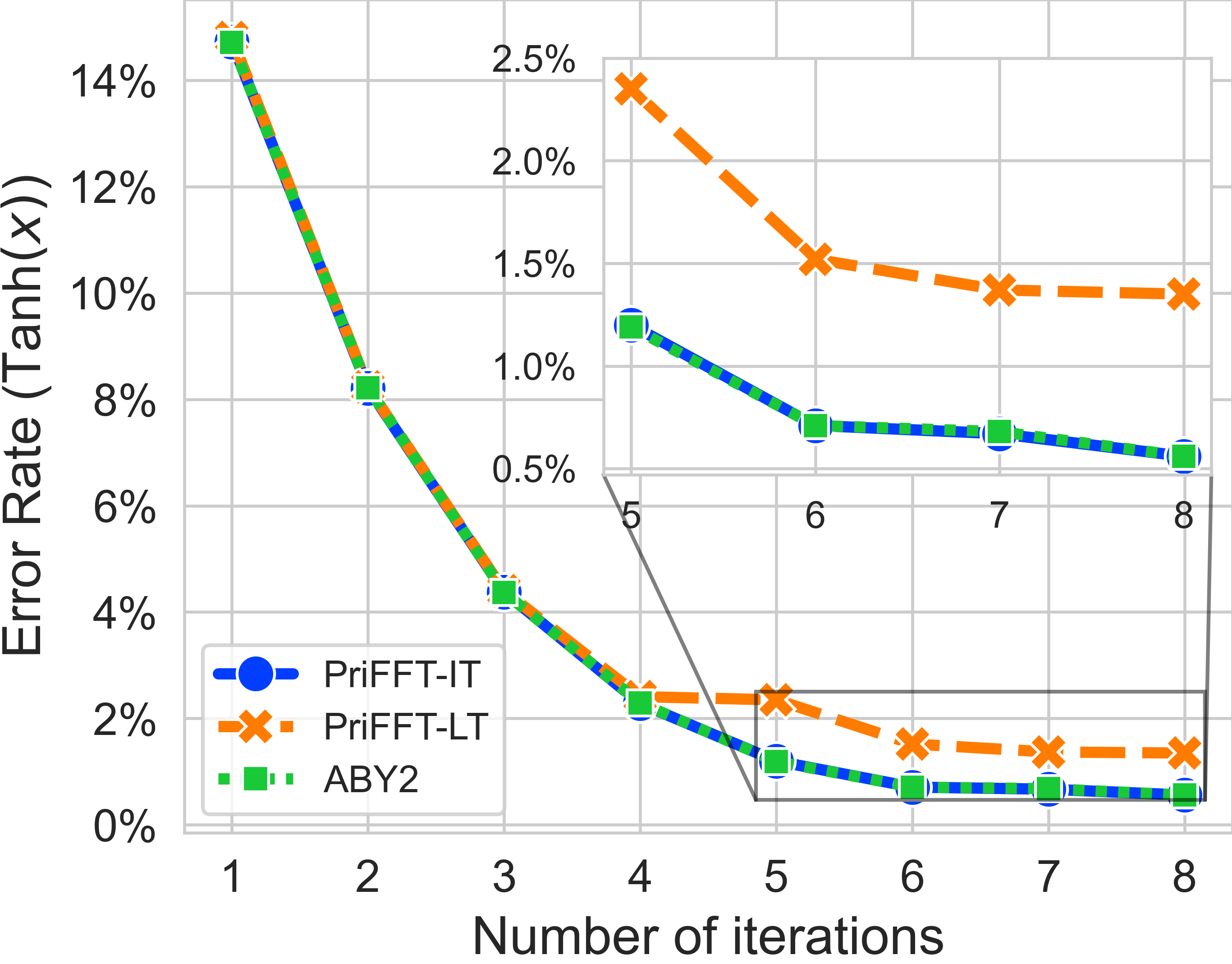}
        \caption{$e^x$}
        \label{fig-evaluation-tanh-exp}
    \end{subfigure}%
    \hspace{0.01\linewidth}
    \begin{subfigure}[b]{0.488\linewidth}
        \centering
        \includegraphics[width=\linewidth]{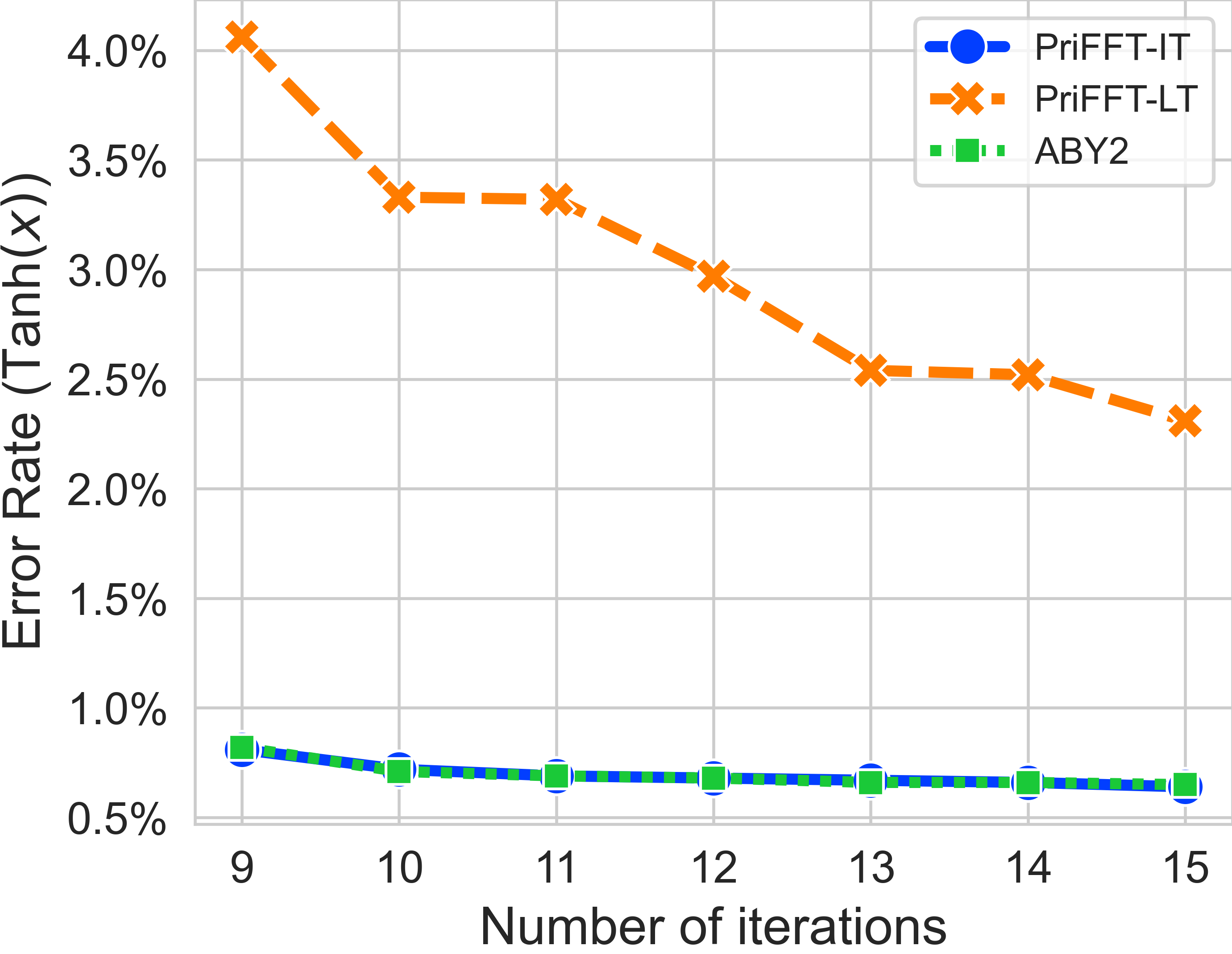}
        \caption{$1/x$}
        \label{fig-evaluation-tanh-recip}
    \end{subfigure}
    \caption{The impact of $m_{\exp}$ and $m_{\mathrm{recip}}$ on the error rates of the secure $\mathrm{Tanh}(x)$.}
    \label{fig-evaluation-tanh}
\end{figure}

Fig.~\ref{fig-evaluation-tanh} provides the impact of $m_{\exp}$ and $m_{\mathrm{recip}}$ on the error rates of $\mathrm{Tanh}(x)$.
The implementation of the secure $\mathrm{Tanh}(x)$ involves the secure $e^x$ and the reciprocal calculation.
Fig.~\ref{fig-evaluation-tanh-exp} presents the error rate of the secure $\mathrm{Tanh}(x)$ under different $m_{\exp}$ when $m_{\mathrm{recip}} = 15$.
Fig.~\ref{fig-evaluation-tanh-exp} presents the error rate of the secure $\mathrm{Tanh}(x)$ under different $m_{\mathrm{recip}}$ when $m_{\exp} = 8$.
Although the error rate of the secure reciprocal calculation is higher than that of the secure $e^x$, the secure reciprocal calculation has a minor effect on the error rate of the secure $\mathrm{Tanh}(x)$ compared to the secure $e^x$.

\begin{figure}
    \centering
    \includegraphics[width=0.96\linewidth]{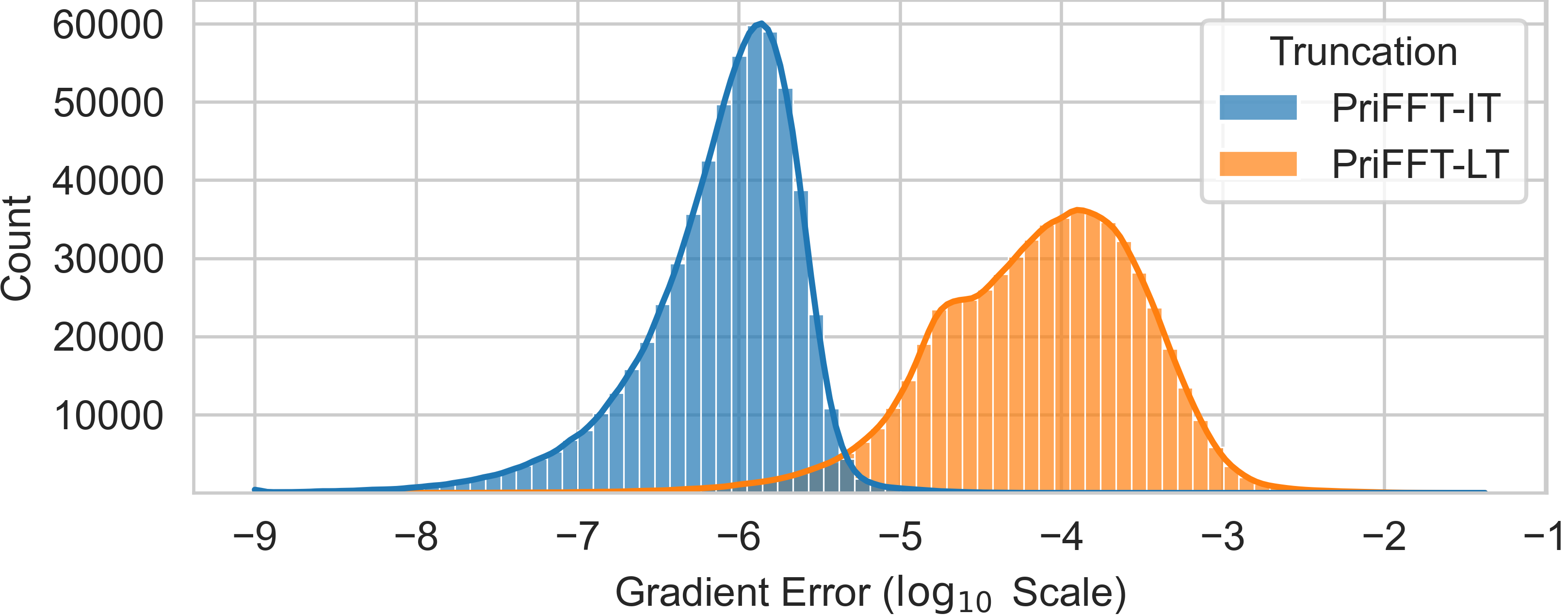}
    \caption{The distribution of gradient error.}
    \label{fig-evaluation-gradient-error}
\end{figure}

We further analyze the impact of different truncation methods on the model gradient error.
Fig.~\ref{fig-evaluation-gradient-error} compares the error of weight gradients in the pooler layer generated by 32 samples with different truncation methods.
PriFFT-IT achieves smaller errors, owing to its higher accuracy in secure computation.
However, most of the errors in PriFFT-LT lie between $10^{-3}$ and $10^{-6}$, which is still minor compared to the gradients.
Therefore, it is more practical to employ LT to reduce communication consumption during the early stages of fine-tuning and subsequently apply IT to enhance calculation accuracy in the later stages of fine-tuning.

\section{Conclusion}
\label{section-conclusion}
In this paper, we first discuss the privacy concern in federated fine-tuning, i.e., both LLMs' parameters and uploaded gradients would cause privacy leakage.
To solve the above privacy problem, we present PriFFT to implement federated fine-tuning of LLMs, while protecting both model parameters and uploaded gradients.
In PriFFT, clients and the server additively share LLM parameters and the inference results, while the privacy-preserving fine-tuning is implemented based on the shared values.
Each party cannot directly access the plaintext LLM parameters and clients' gradients in the privacy-preserving fine-tuning.
Due to the large amount of LLM parameters, fine-tuning LLMs on shared values requires substantial computation and communication resources.
Therefore, we propose several FSS-based secure protocols and the hybrid secret sharing combining ASS and FSS to implement the privacy-preserving federated fine-tuning and reduce the overhead.
We provide detailed overhead and secure analysis for PriFFT. 
Besides, we evaluate PriFFT with privacy-preserving federated fine-tuning implemented by existing mechanisms.
Evaluation results show that PriFFT realizes the same fine-tuned model accuracy while significantly reducing the communication overhead and the execution time.

\bibliographystyle{IEEEtran}
\bibliography{bibliography}

\clearpage
\begin{appendices}
\section{Preliminary Supplement}
\subsection{Arithmetic Secret Sharing}
\label{appendix-ASS}
\subsubsection{Sharing Semantics}
The server and each client in the proposed mechanism works on  two-party computation model where parties $\{P_0, P_1\}$ are connected by a  bidirectional synchronous channel. Specifically, we consider the following sharing semantics:
\begin{itemize}
    \item \textbf{Shared Values.} For two $l$-bit shares $\left \langle x \right \rangle_0$ and $\left \langle x \right \rangle_1$ of $x$, we have $\left \langle x \right \rangle_0 + \left \langle x \right \rangle_1 \equiv x \; (\mathrm{mod} \; 2^l)$ with $\left \langle x \right \rangle_0, \left \langle x \right \rangle_1 \in \mathbb{Z}_{2^l}$.

    \item \textbf{Sharing.} The secret owner $P_b$ sample random value $r \in \mathbb{Z}_{2^l}$, sets $\left \langle x \right \rangle_b = x - r$, and sends $\left \langle x \right \rangle_{1-b} = r$ to $P_{1-b}$.

    \item \textbf{Restore.} When $P_b$ needs to restore $x$, $P_{1-b}$ sends $\left \langle x \right \rangle_{1-b}$ to $P_b$ who restores $x = \left \langle x \right \rangle_{0} + \left \langle x \right \rangle_{1}$.
\end{itemize}

\subsubsection{Operations}
The addition and multiplication of shares are evaluated as follows:

\begin{itemize}
    \item \textbf{Addition.} Given shares $\left \langle x \right \rangle_b$ and $\left \langle y \right \rangle_b$, $P_b$ locally computes $\left \langle x + y \right \rangle_b = \left \langle x \right \rangle_b + \left \langle y \right \rangle_b$.

    \item \textbf{Multiplication.} The multiplication of shares depends on Beaver triples $\{A,B,C\} \subseteq \mathbb{Z}_{2^l}$ generated in the offline phase such that $\left \langle C \right \rangle_b = \left \langle A \cdot B \right \rangle_b$. When $P_b$ computes $\left \langle x \cdot y \right \rangle_b$, $P_b$ sets $\left \langle E \right \rangle_b = \left \langle x \right \rangle_b - \left \langle A \right \rangle_b$ and $\left \langle F \right \rangle_b = \left \langle y \right \rangle_b - \left \langle B \right \rangle_b$. Both parties restore $E$ and $F$ by exchanging corresponding shares. Parties $P_b$ computes $\left \langle x \cdot y \right \rangle_b = b \cdot E \cdot F + F \cdot \left \langle A \right \rangle_b + E \cdot \left \langle B \right \rangle_b + \left \langle C \right \rangle_b$.
\end{itemize}

\subsection{Function Secret Sharing}
\label{appendix-fss}
\begin{definition}[\textbf{FSS: syntax~\cite{boyle2015function}}] A 2-party FSS scheme is a pari of algorithms $(\mathrm{Gen, Eval})$ such that:
    \begin{itemize}
        \item $\mathrm{Gen}(1^\lambda, \hat{f})$ is a probabilistic polynomial-time (PPT) key generation algorithm that outputs a pair of keys $\left(k_0, k_1\right)$ where $\lambda$ is a security parameter and $\hat{f} \in \left\{ 0,1 \right\}^*$ is a description of a function $f$. The function description $\hat{f}$ is assumed to contain descriptions of input and output groups $\mathbb{G}^{\mathrm{in}}$, $\mathbb{G}^{\mathrm{out}}$.

        \item $\mathrm{Eval}(\sigma, k_{\sigma}, x)$ is a polynomial-time evaluation algorithm that outputs a group element $y_{\sigma} \in \mathbb{G}^{\mathrm{out}}$ (the value of $f_\sigma (x)$) where $\sigma \in \left\{ 0,1 \right\}$ (party index), $k_\sigma$ (key defining $f_\sigma: \mathbb{G}^{\mathrm{in}} \rightarrow \mathbb{G}^{\mathrm{out}}$) and $x \in \mathbb{G}^{\mathrm{in}}$ (input for $f_\sigma$).
    \end{itemize}
\end{definition}

\begin{definition}[Offset function family and FSS gates~\cite{boyle2019secure}] Assume that $\mathcal{G} = \left\{ g: \mathbb{G}^{\mathrm{in}} \rightarrow \mathbb{G}^{\mathrm{out}} \right\}$ is a computation gate parameterized by input and output groups $\mathbb{G}^{\mathrm{in}}$ and $\mathbb{G}^{\mathrm{out}}$, the offset function family $\hat{\mathcal{G}}$ of $\mathcal{G}$ is given by
\begin{align}
\nonumber
\hat{\mathcal{G}} := \left\{ 
\begin{array}{l|l}
g^{\left[r^{\text{in}}, r^{\text{out}}\right]}: \mathbb{G}^{\text{in}} \rightarrow \mathbb{G}^{\text{out}} & 
\begin{aligned}
& g: \mathbb{G}^{\text{in}} \rightarrow \mathbb{G}^{\text{out}} \in \mathcal{G}, \\
& r^{\text{in}} \in \mathbb{G}^{\text{in}}, \\
& r^{\text{out}} \in \mathbb{G}^{\text{out}}
\end{aligned}
\end{array}
\right\},
\end{align}
where $g^{ \left[ r^{\mathrm{in}}, r^{\mathrm{out}}
 \right] } (x) := g(x - r^{\mathrm{in}}) + r^{\mathrm{out}}$ and $g^{ \left[ r^{\mathrm{in}}, r^{\mathrm{out}} \right] }$ contains and explicit description of $r^{\mathrm{in}}, r^{\mathrm{out}}$.
\end{definition}

We consider a two-party computation (2PC) between each client and the server with a trusted dealer, secure against a probabilistic polynomial-time semi-honest adversary that passively corrupts one of the parties without deviating from the protocol.
In the offline phase, the trusted dealer distributes input-independent correlated randomness to both parties.
Using this pre-shared randomness, the parties execute the 2PC protocol in the online phase to securely compute the desired function.

\section{Protocol Implementation with FSS}
\subsection{Multiplication}
\label{appendix-mul}
Matrix multiplication is a fundamental calculation in the fine-tuning of LLMs, which implements the fully connected layers and convolutional layers.
PriFFT implements matrix multiplication based on the arithmetic secret sharing in Section~\ref{preliminary-ass}. 
Specifically, given matrix $X \in \mathbb{Z}_{2^l}^{u \times v}$, $Y \in \mathbb{Z}_{2^l}^{v \times w}$ and shares $\left \langle X \right \rangle_b$, $\left \langle Y \right \rangle_b$, computing $\left \langle X \cdot Y \right \rangle_b$ requires generating corresponding dimension Beaver triples $\{C \in \mathbb{Z}_{2^l}^{u \times w}, A \in \mathbb{Z}_{2^l}^{u \times v}, B \in \mathbb{Z}_{2^l}^{v \times w} \}$ in the offline phase such that $\left \langle C \right \rangle_b = \left \langle A \cdot B \right \rangle_b$.
When party $P_b$ computes $\left \langle X \cdot Y \right \rangle_b$, $P_b$ sets $\left \langle E \right \rangle_b = \left \langle X \right \rangle_b - \left \langle A \right \rangle_b$ and $\left \langle F \right \rangle_b = \left \langle Y \right \rangle_b - \left \langle B \right \rangle_b$.
Both parties restore $E$ and $F$ by exchanging corresponding shares. 
Party $P_b$ computes $\left \langle X \cdot Y \right \rangle_b = b \cdot E \cdot F + \left \langle A \right \rangle_b \cdot F + E \cdot \left \langle B \right \rangle_b + \left \langle C \right \rangle_b).$
The communication happens on exchanging $\{\left \langle E \right \rangle_b, \left \langle F \right \rangle_b\}$, and the communication overhead is $\left( 2luv+ 2lvs \right)$ bits.

We denote the multiplication function gate with FSS by $\hat{\mathcal{G}}_{\mathrm{mul}}$ and the offset functions by:
\begin{align}
    \label{equation-multiplication}
    & \hat{g}_{\mathrm{mul}, l, l}^{\left[ r^{\mathrm{in}_1}, r^{\mathrm{in}_2}, r^{\mathrm{out}} \right]} (\hat{x}, \hat{y}) \nonumber \\
    = & g_{\mathrm{mul}, l, l} (\hat{x} - r^{\mathrm{in}_1}, \hat{y} - r^{\mathrm{in}_2}) + r^{\mathrm{out}} \mod 2^l \nonumber \\
    = &  (\hat{x} - r^{\mathrm{in}_1}) \cdot (\hat{y} - r^{\mathrm{in}_2}) + r^{\mathrm{out}} \mod 2^l \nonumber \\
    = &  \hat{x} \hat{y} - \hat{x} r^{\mathrm{in}_2} - \hat{y} r^{\mathrm{in}_1} + r^{\mathrm{in}_1} r^{\mathrm{in}_2} + r^{\mathrm{out}} \mod 2^l.
\end{align}

\begin{algorithm}
    \caption{Multiplications ($x \cdot y$) with FSS in the offline stage $\mathrm{Gen}^{\mathrm{mul}}_{l,l} (r^{\text{in}_1},r^{\mathrm{in}_2}, r^{\text{out}})$ .}
    \label{alg-mul-offline}
    \textbf{Requires: } Input offsets $\{ r^{\text{in}_1}, r^{\text{in}_2} \}$, output offset $r^{\text{out}}$.\\
    \textbf{Output:} Evaluation keys $(k_0, k_1)$. \\
    Sample random $r^{\text{in}_1}_0$, $r^{\text{in}_1}_1$, $r^{\text{in}_2}_0$, $r^{\text{in}_2}_1$, $r^{\text{out}}_0$, $r^{\text{out}}_1 \in \mathbb{Z}_{2^l}$, s.t., $r^{\text{in}_1}_0+r^{\text{in}_1}_1=r^{\text{in}_1} \mod 2^l $, $r^{\text{in}_2}_0+r^{\text{in}_2}_1=r^{\text{in}_2} \mod 2^l $, $r^{\text{out}}_0+r^{\text{out}}_1=r^{\text{out}} \mod 2^l$; \\
    Sample random $q_0, q_1, \in \mathbb{Z}_{2^l}$, s.t., $q_0 + q_1 = r^{\text{in}_1} r^{\text{in}_2} \mod 2^l$; \\
    For $b\in \{0,1\}$, let $k_b = r^{\text{in}_1}_b \| r^{\text{in}_2}_b \| q_b \| r^{\text{out}}_b$ ;\\
    \textbf{return} $(k_0, k_1)$
\end{algorithm}

\begin{algorithm}
    \caption{Multiplications ($x \cdot y$) with FSS in the online stage $\mathrm{Eval}^{\mathrm{mul}}_{l,l}(b, k_b, \left \langle \hat{x} \right \rangle_b, \left \langle \hat{y} \right \rangle_b)$}
    \label{alg-mul-online}
    \textbf{Requires: } Evaluation keys $(k_0, k_1)$, additive shares $(\left \langle \hat{x} \right \rangle_0, \left \langle \hat{x} \right \rangle_1)$ of $\hat{x}$ with offset $r^{\text{in}_1}$, additive shares $(\left \langle \hat{y} \right \rangle_0, \left \langle \hat{y} \right \rangle_1)$ of $\hat{y}$ with offset $r^{\text{in}_2}$, scale factor $s$.\\
    \textbf{Output:} Shares of $(\hat{x}-r^{\text{in}_1})(\hat{y}-r^{\text{in}_2}) + r^{\text{out}}$.\\
    Restore $\hat{x} = x + r^{\mathrm{in}_1}$ from shares of $\hat{x}$ ;\\
    Restore $\hat{y} = y + r^{\mathrm{in}_2}$ from shares of $\hat{y}$ ;\\
    \For{$b \in \{0, 1\}$}{
        Parse $k_b = r^{\text{in}_1}_b \| r^{\text{in}_2}_b \| q_b \| r^{\text{out}}_b$;\\
        $\left \langle \hat{t} \right \rangle_b =  \mathrm{Trunc}(b \hat{x} \hat{y} - \hat{x} r^{\mathrm{in}_2}_b - \hat{y} r^{\mathrm{in}_1}_b + q_b, s)  + r^{\text{out}}_b \mod 2^l;$
    }
    \textbf{return} $(\left \langle \hat{t} \right \rangle_0, \left \langle \hat{t} \right \rangle_1)$
\end{algorithm}

Algorithm~\ref{alg-mul-offline} presents the offline stage of the secure multiplications with FSS.
Algorithm~\ref{alg-mul-online} presents the online stage of the secure multiplications function with FSS.
If $r^{\mathrm{out}}_b$ is not provided, Algorithm~\ref{alg-mul-online} produces additive shares as the result.

\subsection{Truncation}
\label{appendix-truncation}
The scale of multiplication results between shared values would become $2^{2s}$.
The truncation $\mathrm{Trunc}(\left \langle x \right \rangle_{b}, s)$ reduces the scale from $2^{2s}$ to $2^{s}$.
A straightforward implementation is local truncations (LT), i.e., each party directly divides $\left \langle x \right \rangle_{b}$ by $2^s$.
LT produces results with errors when the sum of shares wraps around the ring size.
We introduce the truncation method proposed in CrypTen~\cite{Crypten} as interactive truncations (IT).
$\mathrm{Trunc}(\left \langle x \right \rangle_{b}, s)$ with IT invokes the communication of $l$ bits in the online phase.
Evaluation results in Section~\ref{section-evaluation-model-fine-tuning} show that the difference in accuracy between models fine-tuning with IT and LT is insignificant.
We propose to reduce the communication overhead in the early stages of fine-tuning with LT.
When the model tends to converge, fine-tuning with IT is applied to improve the model accuracy further.
Unless otherwise specified, the communication overhead discussed in Section~\ref{section-mechanism} considers the iterative truncation.

\subsection{Square Function and Power Function}
\label{appendix-power}
We define the square function gate $\mathcal{G}_{\mathrm{square}}$ as the family of functions $g_{\mathrm{square}, l}: \mathbb{Z}_{2^l} \rightarrow \mathbb{Z}_{2^l}$ with input group $\mathbb{G}^{\mathrm{in}} = \mathbb{Z}_{2^l}$, output group $\mathbb{G}^{\mathrm{out}} = \mathbb{Z}_{2^l}$, and $g_{\mathrm{square}, l} (x) := x^2$.

We denote the square function gate with FSS by $\hat{\mathcal{G}}_{\mathrm{squar}}$ and the offset functions by:
\begin{align}
    \label{equation-square}
    \hat{g}_{\mathrm{square}, l}^{\left[ r^{\mathrm{in}}, r^{\mathrm{out}} \right]} (\hat{x})  & = g_{\mathrm{square}, l} (\hat{x} - r^{\mathrm{in}}) + r^{\mathrm{out}} \mod 2^l \nonumber \\
    & = (\hat{x} - r^{\mathrm{in}}) \cdot (\hat{x} - r^{\mathrm{in}}) + r^{\mathrm{out}} \mod 2^l \nonumber \\
    & = \hat{x}^2 - 2 \hat{x} r^{\mathrm{in}} + (r^{\mathrm{in}})^2 + r^{\mathrm{out}} \mod 2^l.
\end{align}

\begin{algorithm}
    \caption{Square function ($x^2$) with FSS in the offline stage $\mathrm{Gen}^{\mathrm{square}}_{l} (r^{\text{in}}, r^{\text{out}})$ .}
    \label{alg-square-offline}
    \textbf{Requires: } Input offset $r^{\text{in}}$, output offset $r^{\text{out}}$.\\
    \textbf{Output:} Evaluation keys $(k_0, k_1)$. \\
    Sample random $r^{\text{in}}_0$, $r^{\text{in}}_1$, $r^{\text{out}}_0$, $r^{\text{out}}_1 \in \mathbb{Z}_{2^l}$, s.t., $r^{\text{in}}_0+r^{\text{in}}_1=r^{\text{in}} \mod 2^l $, $r^{\text{out}}_0+r^{\text{out}}_1=r^{\text{out}} \mod 2^l$; \\
    Sample random $q_0$, $q_1 \in \mathbb{Z}_{2^l}$, s.t., $q_0 + q_1 = (r^{\text{in}})^2 \mod 2^l$; \\
    For $b\in \{0,1\}$, let $k_b = r^{\text{in}}_b \| q_b \| r^{\text{out}}_b$ ;\\
    \textbf{return} $(k_0, k_1)$
\end{algorithm}

Algorithm~\ref{alg-square-offline} presents the offline stage of the secure square function with FSS.
In the offline stage, the input and output offset are split into two additive shares $\left\{ r^{\mathrm{in}}_0, r^{\mathrm{in}}_1 \right\}$ and $\left\{ r^{\mathrm{out}}_0, r^{\mathrm{out}}_1 \right\}$ (line 3).
Besides, the multiplication result of $(r^{\mathrm{in}})^2$ is also spilt into two additive shares $\left\{ q_0, q_1 \right\}$ (line 4).
Since the computation is performed on a $\mathbb{Z}_{2^l}$ ring, the bitwidth of $(r^{\mathrm{in}})^2$, $q_0$, and $q_1$ are $l$.
Therefore, the key size of the secure square function with FSS is $3l$ bits for each party, which is the same as the implementation based on ASS where keys include Beaver triples $\left\{ A, B, C \right\}$.

\begin{algorithm}
    \caption{Square function ($x^2$) with FSS in the online stage $\mathrm{Eval}^{\mathrm{square}}_{l}(b, k_b, \left \langle \hat{x} \right \rangle_b)$}
    \label{alg-square-online}
    \textbf{Requires: } Evaluation keys $(k_0, k_1)$, additive shares $(\left \langle \hat{x} \right \rangle_0, \left \langle \hat{x} \right \rangle_1)$ of $\hat{x}$ with offset $r^{\text{in}}$, scale factor $s$.\\
    \textbf{Output:} Shares of $(\hat{x}-r^{\text{in}})^2 + r^{\text{out}}$.\\
    Restore $\hat{x} = x + r^{\mathrm{in}}$ from shares of $\hat{x}$ ;\\
    \For{$b \in \{0, 1\}$}{
        Parse $k_{b} = r^{\text{in}}_b \| q_b \| r^{\text{out}}_b$;\\
        $\left \langle \hat{t} \right \rangle_b =  \mathrm{Trunc}(b\hat{x}^2 - 2 r^{\mathrm{in}}_b \hat{x} + q_b, s) + r^{\text{out}}_b \mod 2^l;$
    }
    \textbf{return} $(\left \langle \hat{t} \right \rangle_0, \left \langle \hat{t} \right \rangle_1)$
\end{algorithm}

Algorithm~\ref{alg-square-online} presents the online stage of the secure square function with FSS.
Each party exchanges their shares of $\hat{x}$ to restore $\hat{x} = x + r^{\mathrm{in}}$ (line 3).
The shares of results are computed according to Equation~\eqref{equation-square} (line 6)
Since each party cannot obtain $r^{\mathrm{in}}$, restoring $\hat{x}$ with the input offset does not leak information about $x$.
The input $\hat{x}$ of Algorithm~\ref{alg-square-online} is masked by the input offset $r^{\mathrm{in}}$.
Meanwhile, $\hat{x}$ is the output of a previous FSS gate with an output offset $r^{\mathrm{out}'}$.
The correctness of the square function gate with FSS lies in setting $r^{\mathrm{in}}$ equal to the output offset of the previous FSS gate, $r^{\mathrm{out}'}$, in the offline stage.

Generally, $\mathrm{Gen}^{\mathrm{square}}_{l}$ is the key generation in the offline stage and outputs keys $(k_0, k_1)$ belonging to the server and the client.
In the online stage, $\mathrm{Eval}^{\mathrm{square}}_{l}$ uses $(k_0, k_1)$ and the shared inputs of $x$ to calculate the shares of $x^2$.

The power function is implemented  based on exponentiation by squaring, i.e.,
\begin{equation}
    f(x, n) = \begin{cases}
    1 & \text{if } n=0; \\
    x & \text{if } n=1; \\
    (f(x, n / 2))^2 & \text{if } n \text{ is even}; \\
    x \times(f(x,(n-1) / 2))^2 & \text{if } n \text{ is odd}.
    \end{cases}
    \label{equation-power-function}
\end{equation}
For convenience, we consider the case where $n=2^m$ to describe the secure power function.
In the case of $n \neq 2^m$, the commutation can be implemented by combining Equation~\eqref{equation-power-function} and the following proposed algorithms.

\begin{algorithm}
    \caption{Power function ($x^{2^m}$) with FSS in the offline stage $\mathrm{Gen}^{\mathrm{power}}_{l} (r^{\text{in}}, r^{\text{out}}, \mathbf{r}, 2^m)$.}
    \label{alg-power-offline}
    \textbf{Requires: } Input offset $r^{\text{in}}$, output offset $r^{\text{out}}$, intermediate offsets $\mathbf{r}=\{ r^{(1)}, \cdots, r^{(m)}\}$, power $2^m$.\\
    \textbf{Output:} Evaluation keys $(k_0, k_1)$. \\
    Sample random $r^{\text{in}}_0$, $r^{\text{in}}_1$, $r^{\text{out}}_0$, $r^{\text{out}}_1 \in \mathbb{Z}_{2^l}$, s.t., $r^{\text{in}}_0+r^{\text{in}}_1=r^{\text{in}} \mod 2^l $, $r^{\text{out}}_0+r^{\text{out}}_1=r^{\text{out}} \mod 2^l$; \\
    \ForEach{$r_i \in  \{r^{(1)}, \cdots r^{({m})}\}$}{
        Sample random $r^{(i)}_0$, $r^{(i)}_1$, $q^{(i)}_0$, $q^{(i)}_1 \in \mathbb{Z}_{2^l} $, s.t., $r^{(i)}_0+r^{(i)}_1=r^{(i)} \mod 2^l$, $q^{(i)}_0+q^{(i)}_1=\left[r^{(i)}\right]^2 \mod 2^l$; \\
    }
    For $b\in \{0,1\}$, let $k_b = r^{\text{in}}_b \| r^{(1)}_b \| q^{(1)}_b \| \cdots \| r^{(m)}_b \| q^{(m)}_b \| r^{\text{out}}_b$; \\
    \textbf{return} $(k_0, k_1)$
\end{algorithm}

\begin{algorithm}
    \caption{Power function ($x^{2^m}$) with FSS in the online stage $\mathrm{Eval}^{\mathrm{power}}_{l}(b, k_b, \left \langle \hat{x} \right \rangle_b, 2^m)$}
    \label{alg-power-online}
    \textbf{Requires: } Evaluation keys $(k_0, k_1)$, additive shares $(\left \langle \hat{x} \right \rangle_0, \left \langle \hat{x} \right \rangle_1)$ of $\hat{x}$ with offset $r^{\text{in}}$, scale factor $s$.\\
    \textbf{Output:} Shares of $(\hat{x}-r^{\text{in}})^2 + r^{\text{out}}$.\\
    \For{$b \in \{0, 1\}$}{
        Parse $k_b = r^{\text{in}}_b \| r^{(1)}_b \| q^{(1)}_b \| \cdots \| r^{(m)}_b \| q^{(m)}_b \| r^{\text{out}}_b$;\\
        $\left \langle \hat{t}_{0} \right \rangle_b = \left \langle \hat{x} \right \rangle_b + r_n^{\mathrm{in}};$\\
        \For{$i=1 \: \mathrm{to} \:m$}{
        \If{$i < m_{\exp}$}{
            $ k_b^{\mathrm{square}} = r^{(i)}_b \| q^{(i)}_b \| r^{(i+1)}_b$ ; \\
        }
        \Else{
            $ k_b^{\mathrm{square}} = r^{(i)}_b \| q^{(i)}_b \| r^{\text{out}}_b$ ;
        }
        $\left \langle \hat{t}_i \right \rangle_b = \mathrm{Eval}^{\mathrm{square}}_{l}(b, k_b^{\mathrm{square}}, \left \langle \hat{t}_{i-1} \right \rangle_b)$ \\
    }
    }
    \textbf{return} $(\left \langle \hat{t_m} \right \rangle_0, \left \langle \hat{t_m} \right \rangle_1)$
\end{algorithm}

\subsection{Natural Exponential Function}
\label{appendix-exp}
\begin{algorithm}
    \caption{Natural exponential function ($e^x$) with FSS in the offline stage $\mathrm{Gen}^{\mathrm{exp}}_{l} (r^{\text{in}}, r^{\text{out}}, \mathbf{r})$}
    \label{alg-exp-offline}
    \textbf{Requires: } Input offset $r^{\text{in}}$, output offset $r^{\text{out}}$, iteration rounds $m_{\exp}$, intermediate offsets $ \mathbf{r} = \{r^{(1)}, \cdots, r^{(m_{\exp})}\}$.\\
    \textbf{Output:} Evaluation keys $(k_0, k_1)$. \\
    Sample random $r^{\text{in}}_0$, $r^{\text{in}}_1$, $r^{\text{out}}_0$, $r^{\text{out}}_1 \in \mathbb{Z}_{2^l}$, s.t., $r^{\text{in}}_0 + r^{\text{in}}_1 = r^{\text{in}} \mod 2^l$, $r^{\text{out}}_0 + r^{\text{out}}_1 = r^{\text{out}} \mod 2^l$; \\
    \ForEach{$r_i \in  \{r^{(1)}, \cdots r^{({m_{\exp}})}\}$}{
        Sample random $r^{(i)}_0$, $r^{(i)}_1$, $q^{(i)}_0$, $q^{(i)}_1 \in \mathbb{Z}_{2^l} $, s.t., $r^{(i)}_0+r^{(i)}_1=r^{(i)} \mod 2^l$, $q^{(i)}_0+q^{(i)}_1=\left[r^{(i)}\right]^2 \mod 2^l$; \\
    }
    For $b\in \{0,1\}$, let $k_b = r^{\text{in}}_b \| r^{(1)}_b \| q^{(1)}_b \| \cdots \| r^{(m_{\exp})}_b \| q^{(m_{\exp})}_b \| r^{\text{out}}_b$; \\
    \textbf{return} $(k_0, k_1)$
\end{algorithm}

We present the offline stage of the natural exponential function gate with FSS in Algorithm~\ref{alg-exp-offline}.
The implementation computes $(1 + \hat{x} / 2^{m_{\exp}})$ locally and then performs $m_{\exp}$ times square functions with FSS.
We introduce additional offsets $\left\{ r^{(1)} \cdot, r^{(m_{\exp})} \right\}$ to protect the intermediate results of square functions.
The key generation only shares $r^{(i)}$ and $\left[r^{(i)}\right]^2$ instead of directly applying $m_{\exp}$ times $\mathrm{Gen}^{\mathrm{square}}_l (r^{(i)}, r^{(i+1)})$.

Algorithm~\ref{alg-exp-online} presents the online stage of the natural exponential function with FSS.
We first compute $\left( 1 + x / 2^{m_{\exp}} \right)$ by the masked input $\hat{x}$ (line 6).
The term $2^s b$ refers to 1 on the ring $\mathbb{Z}_{2^l}$ in the above equation.
The value of $x / 2^{m_{\exp}}$ is computed by a truncation function $\mathrm{Trunc}(\hat{x} - r^{\text{in}}_b, m_{\exp})$.
The intermediate result $t^{(i)}_b$ is masked by off offset $r_b^{(1)}$.
In each square computation, both parties compute the square results with corresponding offsets (line 11).

\begin{algorithm}
    \caption{Natural exponential function ($e^x$) with FSS in the online stage $\mathrm{Eval}^{\mathrm{exp}}_{l}(b, k_b, \left \langle \hat{x} \right \rangle_b)$}
    \label{alg-exp-online}
    \textbf{Requires: } Evaluation keys $(k_0, k_1)$, additive shares $(\left \langle \hat{x} \right \rangle_0, \left \langle \hat{x} \right \rangle_1)$ of $\hat{x}$ with offset $r^{\text{in}}$, scale factor $s$.\\
    \textbf{Output:} Shares of $e^{\hat{x} - r^{\text{in}}} + r^{\text{out}}$.\\
    Restore $\hat{x} = x + r^{\mathrm{in}}$ from shares of $\hat{x}$; \\
    \For{$b \in \{0, 1\}$}{
        Parse $k_b = r^{\text{in}}_b \| r^{(1)}_b \| q^{(1)}_b \| \cdots \| r^{(m_{\exp})}_b \| q^{(m_{\exp})}_b \| r^{\text{out}}_b$;\\
        $ \left \langle \hat{t}_0 \right \rangle_b =  \mathrm{Trunc}(\hat{x} - r^{\text{in}}_b, m_{\exp}) + 2^s b + r^{(1)}_b \mod 2^l$;\\
        \For{$i=1 \: \mathrm{to} \:m_{\exp}$}{
        \If{$i < m_{\exp}$}{
                $ k_b^{\mathrm{square}} = r^{(i)}_b \| q^{(i)}_b \| r^{(i+1)}_b$ ; \\
            }
            \Else{
                $ k_b^{\mathrm{square}} = r^{(i)}_b \| q^{(i)}_b \| r^{\text{out}}_b$ ;
            }
            $\left \langle \hat{t}_i \right \rangle_b = \mathrm{Eval}^{\mathrm{square}}_{l}(b, k_b^{\mathrm{square}}, \left \langle \hat{t}_{i-1} \right \rangle_b)$ \\
        }
    }
    
    \textbf{return} $(\left \langle \hat{t}_{m_{\exp}} \right \rangle_0, \left \langle \hat{t}_{m_{\exp}} \right \rangle_1)$
\end{algorithm}

\subsection{Reciprocal Operation}
\label{appendix-reciprocal}
Algorithm~\ref{alg-recip-offline} presents the offline stage of the reciprocal operation with FSS.
The key generation of the reciprocal operation is similar to that of the natural exponential function.
The difference is Algorithm~\ref{alg-recip-offline} applies $\mathrm{Gen}^{\exp}_{l}$ to generate secure $e^x$ for the initial values calculation (line 3).
Besides, Algorithm~\ref{alg-recip-offline} generates shares of offsets to support multiplications with FSS (lines 6-7).

\begin{algorithm}
    \caption{Reciprocal ($1/x$) with FSS in the offline stage $\mathrm{Gen}^{\mathrm{recip}}_{l} ( r^{\text{in}}, r^{\text{out}}, \mathbf{r})$}
    \label{alg-recip-offline}
    \textbf{Requires: } Input offset $r^{\text{in}}$, output offset $r^{\text{out}}$, iteration rounds $m_{\mathrm{recip}}$, $e^x$ offsets $\mathbf{r}^{\exp}$, intermediate offsets $\mathbf{r} = \{\mathbf{r}^{\exp} , r^{(1)}, \cdots, r^{(2 m_{\mathrm{recip}})}\}$, security parameter $\lambda$.\\
    \textbf{Output:} Evaluation keys $(k_0, k_1)$. \\
    $(k_0^{\mathrm{exp}}, k_1^{\mathrm{exp}}) = \mathrm{Gen}^{\mathrm{exp}}_{l} (r^{\text{in}}, \emptyset, \mathbf{r}^{\mathrm{exp}})$; \\
    Sample random $r^{\text{in}}_0$, $r^{\text{in}}_1$, $r^{\text{out}}_0$, $r^{\text{out}}_1 \in \mathbb{Z}_{2^l}$, s.t., $r^{\text{in}}_0 + r^{\text{in}}_1 = r^{\text{in}} \mod 2^l$, $r^{\text{out}}_0 + r^{\text{out}}_1 = r^{\text{out}} \mod 2^l$; \\
    \ForEach{$i \in  \{1, \cdots, m_{\mathrm{recip}}\}$}{
        Sample random $r^{(i)}_0$, $r^{(i)}_1$, $q^{(i)}_0$, $q^{(i)}_1$, s.t., $r^{(i)}_0+r^{(i)}_1=r^{(2i-1)} \mod 2^l$, $q^{(i)}_0+q^{(i)}_1=\left[r^{(2i-1)}\right]^2 \mod 2^l$; \\
        Sample random $u^{(i)}_0$, $u^{(i)}_1$, $v^{(i)}_0$, $v^{(i)}_1$, s.t., $u^{(i)}_0+u^{(i)}_1=r^{(2i)} \mod 2^l$, $v^{(i)}_0+v^{(i)}_1=r^{(2i) } r^{\mathrm{in}} \mod 2^l$; \\
    }
    For $b\in \{0,1\}$, let $k_b = k^{\exp}_b \| r^{\text{in}}_b \| r^{(1)}_b \| q^{(1)}_b \| u^{(1)}_b \| v^{(1)}_b \| \cdots \| v^{(m_{\mathrm{recip}})}_b \| r^{\text{out}}_b$; \\
    \textbf{return} $(k_0, k_1)$
\end{algorithm}

Algorithm~\ref{alg-recip-online} presents the online stage of the reciprocal operation with FSS.
Algorithm~\ref{alg-recip-online} firstly calculates the initial value $y_0$ based on the secure $e^x$ protocol (line 6).
The initial value is masked by the offset $r_b^{(1)}$ to support the implementation with FSS.
The reciprocal of the input is obtained through the approximation (lines 7-15).
We compute $x_n^2$ by the FSS square gate (line 13) and $a x_n^2$ by the FSS multiplication gate (line 14).
The intermediate result in each iteration $\hat{t}_b^{(i)}$ is $2 x_n - a x_n^2$ (line 15), where $\hat{t}^{(i-1)}_b - r^{(i-1)}_b$  removes the offset of $\hat{t}^{(i-1)}_b$ so that the offset of $\hat{t}_b^{(i)}$ is $r^{(i-1)}_b$ or $r^{\mathrm{out}}_b$.
When $y_0$ is provided in $\mathrm{Eval}^{\mathrm{recip}}_{l}(b, k_b, \left \langle \hat{x} \right \rangle_b, y_0)$, the initial value calculation can be omitted to reduce the calculation overhead.

\begin{algorithm}
    \caption{Reciprocal ($1/x$) with FSS in the online stage $\mathrm{Eval}^{\mathrm{recip}}_{l}(b, k_b, \left \langle \hat{x} \right \rangle_b)$}
    \label{alg-recip-online}
    \textbf{Requires: } Evaluation keys $(k_0, k_1)$, additive shares $(\left \langle \hat{x} \right \rangle_0, \left \langle \hat{x} \right \rangle_1)$ of $\hat{x}$ with offset $r^{\text{in}}$, initial guess $x_0$.\\
    \textbf{Output:} Shares of $\frac{1}{\hat{x} - r^{\mathrm{in}}} + r^{\text{out}}$.\\
    Restore $\hat{x} = x + r^{\mathrm{in}}$ from shares of $\hat{x}$; \\
    \For{$b \in \{0, 1\}$}{
        Parse $k_b = k^{\exp}_b \| r^{\text{in}}_b \| r^{(1)}_b \| q^{(1)}_b \| u^{(1)}_b \| v^{(1)}_b \| \cdots \| v^{(m_{\mathrm{recip}})}_b \| r^{\text{out}}_b$ \\
        $ \left \langle \hat{t}_0 \right \rangle_b = 3 \cdot \mathrm{Eval}^{\exp}_l(b, k^{\exp}_b, 1-2\left\langle x \right\rangle_b) + 0.003$; \\
        \For{$i=1 \: \mathrm{to} \:m_{\mathrm{recip}}$}{
        $ k_b^{\mathrm{square}} = r^{(i)}_b \| q^{(i)}_b \| u_b^{(i)}$ ;\\
            \If{$i < m_{\mathrm{recip}}$}{
                $ k_b^{\mathrm{mul}} = r_b^{\mathrm{in}} \| u_b^{(i)} \| v_b^{(i)} \| r^{(i+1)}_b$ ;\\
            }
            \Else{
                $ k_b^{\mathrm{mul}} = r_b^{\mathrm{in}} \| u_b^{(i)} \| v_b^{(i)} \| r^{\mathrm{out}}_b$ ;\\
            }
            $\left \langle \hat{t}_i \right \rangle_b = \mathrm{Eval}^{\mathrm{square}}_{l}(b, k_b^{\mathrm{square}}, \left \langle \hat{t}_{i-1} \right \rangle_b)$; \\
            $\left \langle \hat{t}_i \right \rangle_b = \mathrm{Eval}^{\mathrm{mul}}_{l,l}(b, k_b^{\mathrm{mul}}, \hat{x}, \left \langle \hat{t}_i \right \rangle_b)$; \\
            $\left \langle \hat{t}_i \right \rangle_b = 2 \cdot \left[\left \langle \hat{t}_{i-1} \right \rangle_b - r^{(i-1)}_b\right] - \left \langle \hat{t}_i \right \rangle_b \mod 2^l;$ \\
        }
    }
    \textbf{return} $(\left \langle \hat{t}_{m_{\mathrm{recip}}} \right \rangle_0, \left \langle \hat{t}_{m_{\mathrm{recip}}} \right \rangle_1)$
\end{algorithm}

\subsection{Softmax Function}
\begin{algorithm}
    \caption{Softmax function with FSS in the offline stage $\mathrm{Gen}^{\mathrm{softmax}}_{K \times l} (\mathbf{r}^{\text{in}},\mathbf{r}^{\text{out}}, \mathbf{r})$}
    \label{alg-softmax-offline}
    \textbf{Requires: } Input offsets $\mathbf{r}^{\text{in}}$, output offset $\mathbf{r}^{\text{out}}$, intermediate offsets $\mathbf{r}=\{\mathbf{r}^{(1)}, \mathbf{r}^{(2)}, \cdots, \mathbf{r}^{2m_{\mathrm{softmax}}}  \}$ .\\
    \textbf{Output:} Evaluation keys $(k_0, k_1)$. \\
    Sample random $\mathbf{r}^{\text{in}}_0$, $\mathbf{r}^{\text{in}}_1$, $\mathbf{r}^{\text{out}}_0$, $\mathbf{r}^{\text{out}}_1 \in \mathbb{Z}_{K\times2^l}$, s.t., $\mathbf{r}^{\text{in}}_0 + \mathbf{r}^{\text{in}}_1 = \mathbf{r}^{\text{in}} \mod 2^l$, $\mathbf{r}^{\text{out}}_0 + \mathbf{r}^{\text{out}}_1 = \mathbf{r}^{\text{out}} \mod 2^l$; \\
    \ForEach{$i \in  \{1, \cdots, m_{\mathrm{softmax}}\}$}{
        $k_b^{\mathrm{mul}_{i,1}}=\mathrm{Gen}^{\mathrm{mul}}_{K\times l, K \times l} (\mathbf{r}^{\mathrm{in}},\mathbf{r}^{2i-1}, \emptyset)$; \\
        $k_b^{\mathrm{mul}_{i,2}}=\mathrm{Gen}^{\mathrm{mul}}_{K\times l, K \times l} (\mathbf{r}^{\mathrm{2i-1}},\mathbf{r}^{2i}, \emptyset)$ \\
    }
    For $b\in \{0,1\}$, let $k_b = k_b^{\mathrm{mul}_{1,1}} \| k_b^{\mathrm{mul}_{1,2}} \| \cdots \| k_b^{\mathrm{mul}_{m_{\mathrm{softmax}},1}} \| k_b^{\mathrm{mul}_{m_{\mathrm{softmax}},2}}$; \\
    \textbf{return} $(k_0, k_1)$
\end{algorithm}

Algorithm~\ref{alg-softmax-offline} presents the offline stage of the softmax function with FSS.
In the offline phases, the algorithm calls $\mathrm{Gen}_{K \times l,K \times l}^{\mathrm{mul}}$ twice.
The output offsets in $\mathrm{Gen}_{K \times l,K \times l}^{\mathrm{mul}}$ are $\emptyset$, which means the outputs are not masked by the offset and are additive shares.
Since the online stage requires $m_{\mathrm{softmax}}$ rounds, Algorithm~\ref{alg-softmax-offline} produces $m_{\mathrm{softmax}}$ sets of masks for multiplication.

\begin{algorithm}
    \caption{Softmax function with FSS in the online stage $\mathrm{Eval}^{\mathrm{softmax}}_{K \times l}(b, k_b, \left \langle \hat{\mathbf{x}} \right \rangle_b)$}
    \label{alg-softmax-online}
    \textbf{Requires: } Evaluation keys $(k_0, k_1)$, additive shares $(\left \langle \hat{\mathbf{x}} \right \rangle_0, \left \langle \hat{\mathbf{x}} \right \rangle_1)$ of $\hat{\mathbf{x}}$ with offset $\mathbf{r}^{\text{in}}$.\\
    \textbf{Output:} Shares of $\mathrm{Softmax}(\hat{\mathbf{x}} - \mathbf{r}^{\text{in}})$.\\
    \For{$b \in \{0, 1\}$}{
        Parse $k_b = k_b^{\mathrm{mul}_{1,1}} \| k_b^{\mathrm{mul}_{1,2}} \| \cdots \| k_b^{\mathrm{mul}_{m_{\mathrm{softmax}},1}} \| k_b^{\mathrm{mul}_{m_{\mathrm{softmax}},2}}$;\\
        $\sigma_0 \|\sigma_1 = \mathrm{ReLU(a_0 - \left \langle \hat{\mathbf{x}} \right \rangle_b \| \left \langle \hat{\mathbf{x}} \right \rangle_b - a_1)}$;\\
        $\left \langle \hat{\mathbf{x}} \right \rangle_b = \mathrm{Trunc}(\left \langle \hat{\mathbf{x}} \right \rangle_b +\sigma_a - \sigma_b, m_\mathrm{softmax})$ ;\\
        $\left\langle \mathbf{y} \right\rangle_0 = \{b / K, \cdots,b / K\}$;\\
        \For{$i=1 \: \mathrm{to} \:m_{\mathrm{softmax}}$}{
            $ \left\langle \mathbf{t} \right\rangle_b = \mathrm{Eval}^{\mathrm{mul}}_{K \times l,K \times l}(b, k_b^{\mathrm{mul}_{i,1}}, \left\langle \hat{\mathbf{x}} \right\rangle_b, \left \langle \mathbf{y}_{i-1} \right \rangle_b)$; \\
            $\left\langle \mathbf{q} \right\rangle_b = \mathrm{Eval}^{\mathrm{mul}}_{K \times l,K \times l}(b, k_b^{\mathrm{mul}_{i,2}}, \left \langle \mathbf{y}_{i-1} \right \rangle_b, \sum \left\langle \mathbf{t} \right\rangle_b)$;\\
            $\left \langle \mathbf{y}_{i} \right \rangle_b = \left \langle \mathbf{y}_{i-1} \right \rangle_b + \left\langle \mathbf{t} \right\rangle_b - \left\langle \mathbf{q} \right\rangle_b$; \\
            }
        \textbf{return} $(\left \langle \mathbf{y}_{m_{\mathrm{softmax}}} \right \rangle_0, \left \langle \mathbf{y}_{m_{\mathrm{softmax}}} \right \rangle_1)$
    }
\end{algorithm}

Algorithm~\ref{alg-softmax-online} presents the online stage of the softmax function with FSS.
Initially, we use secure ReLU to clip the input $\left \langle \hat{\mathbf{x}} \right \rangle_b$ within the interval $(a_0,a_1)$ (lines 5-6).
Note that the $\left \langle \hat{\mathbf{x}} \right \rangle_b$ is divided by $m_{\mathrm{softmax}}$ through the truncation before the iteration, which can omit the division in each iteration.
The secure ReLU is implemented in CrypTen.
The initial values are set to $1/K$ (line 7).
In each iteration, we first calculate $x \cdot y_{n-1}$ using the multiplication gate (line 9).
Then, $\left\langle x, y_{n-1}\right\rangle \cdot y_{n-1}$ is calculated by the second multiplication gate (line 10).
The $\left\langle \hat{\mathbf{x}} \right\rangle_b$ serves as input repeatedly, necessitating its restoration just a single time in the multiplication gate.
In the same way, $y_i$ appears twice in each iteration, necessitating only a single restoration in each iteration.

\subsection{Sigmoid Function and Hyperbolic Tangent}
\label{appendix-sigmoid-tanh}
\begin{algorithm}
    \caption{Sigmoid function with FSS in the offline stage $\mathrm{Gen}^{\mathrm{sigmoid}}_{l} (r^{\text{in}}, r^{\text{out}}, \mathbf{r}^{\mathrm{exp}}, \mathbf{r}^{\mathrm{recip}}, \mathbf{r})$}
    \label{alg-sigmoid-offline}
    \textbf{Requires: } Input offsets $r^{\text{in}}$, output offset $r^{\text{out}}$, offsets $(\mathbf{r}^{\mathrm{exp}}, \mathbf{r}^{\mathrm{recip}})$ for $\mathrm{Gen}^{\mathrm{exp}}_{l}$ and $\mathrm{Gen}^{\mathrm{recip}}_{l}$, intermediate offsets $\mathbf{r}=\{r^{(1)},\cdots,r^{(7)}\}$, security parameter $\lambda$.\\
    \textbf{Output:} Evaluation keys $(k_0, k_1)$. \\
    $(c_0,c_1) = \mathrm{Gen}_{l}^<(1^\lambda, r^{\mathrm{in}}, 1, \mathbb{Z}_{2^l})$; \\
    $(k_0^{\mathrm{mul}_1}, k_1^{\mathrm{mul}_1}) = \mathrm{Gen}^{\mathrm{mul}}_{l, l} (r^{\mathrm{in}}, r^{(1)}, r^{(2)})$; \\
    $(k_0^{\mathrm{exp}}, k_1^{\mathrm{exp}}) = \mathrm{Gen}^{\mathrm{exp}}_{l} (r^{(2)}, r^{(3)}, \mathbf{r}^{\mathrm{exp}})$; \\
    $(k_0^{\mathrm{recip}}, k_1^{\mathrm{recip}}) = \mathrm{Gen}^{\mathrm{recip}}_{l} (r^{(3)}, \emptyset, \mathbf{r}^{\mathrm{recip}})$; \\
    $(k_0^{\mathrm{mul}_2}, k_1^{\mathrm{mul}_2}) = \mathrm{Gen}^{\mathrm{mul}}_{l, l} (r^{(4)}, r^{(5)}, \emptyset)$; \\
    $(k_0^{\mathrm{mul}_3}, k_1^{\mathrm{mul}_3}) = \mathrm{Gen}^{\mathrm{mul}}_{l, l} (r^{(6)}, r^{(7)}, \emptyset)$; \\
    Sample random $r^{(i)}_b\in \mathbb{Z}_{2^l}$ where $b\in\{0,1\}$ and $i=\{1,4,5,6,7\}$, s.t., $r^{(i)}_0 + r^{(i)}_1 = r^{(i)} \mod 2^l$; \\
    For $b\in \{0,1\}$, let $k_b = k_b^{\mathrm{exp}} \| k_b^{\mathrm{recip}} \| k_b^{\mathrm{mul}_1} \| k_b^{\mathrm{mul}_2} \| k_b^{\mathrm{mul}_3} \| r_b^{(1)} \| r_b^{(4)} \| \cdots \| r_b^{(7)}$; \\
    \textbf{return} $(k_0, k_1)$
\end{algorithm}

Algorithm~\ref{alg-sigmoid-offline} presents the offline stage of the sigmoid function with FSS.
The evaluation of the secure sigmoid function requires three secure multiplications (performed in lines 4, 7, and 8), a secure comparison (in line 3), a secure $e^x$ computation (in line 5), and a secure reciprocal calculation (in line 6). 
In the offline phase, Algorithm~\ref{alg-sigmoid-offline} generates the corresponding keys.
Algorithm~\ref{alg-sigmoid-online} performs the specific calculation of the sigmoid function with the keys generated in the offline.
The algorithm firstly performs a secure comparison between $\left\langle \hat{x} \right\rangle$ and 0 (line 6), subsequently determining the sign of $\left\langle \hat{x} \right\rangle$ (line 7).
Subsequently, we calculate the absolute values $\left\langle \hat{|x|} \right\rangle$ through a secure multiplication involving $\left\langle \hat{x} \right\rangle$ and its sign (line 8).
The sigmoid function with an output of $\left\langle \hat{|x|} \right\rangle$ is calculated by the secure $e^x$ and reciprocal calculation (lines 9-10).
Finally, the algorithm adjusts the sigmoid results according to the sign of $\left\langle \hat{x} \right\rangle$ (line 11).

\begin{algorithm}
    \caption{Sigmoid function with FSS in the online stage $\mathrm{Eval}^{\mathrm{sigmoid}}_{l}(b, k_b, \left \langle \hat{x} \right \rangle_b)$}
    \label{alg-sigmoid-online}
    \textbf{Requires: } Evaluation keys $(k_0, k_1)$, additive shares $(\left \langle \hat{x} \right \rangle_0, \left \langle \hat{x} \right \rangle_1)$ of $\hat{x}$ with offset $r^{\text{in}}$.\\
    \textbf{Output:} Shares of $\mathrm{Sigmoid}(\hat{x} - r^{\text{in}})$.\\
    \For{$b \in \{0, 1\}$}{
        Parse $k_b = k_b^{\mathrm{exp}} \| k_b^{\mathrm{recip}} \| k_b^{\mathrm{mul}_1} \| k_b^{\mathrm{mul}_2} \| k_b^{\mathrm{mul}_3} \| r_b^{(1)} \| r_b^{(4)} \| \cdots \| r_b^{(7)}$;\\
        Restore $\hat{x} = x + r^{\mathrm{in}}$ from shares of $\hat{x}$;\\
        $\left\langle p \right\rangle_b = \mathrm{Eval}_l^{<}(b, c_b, \hat{x})$;\\
        $\left\langle \hat{t} \right\rangle_b = 1-2 \cdot \left\langle p \right\rangle_b  + r_b^{(1)}$ ; \\
        $\left\langle \hat{t} \right\rangle_b = \mathrm{Eval}_{l, l}^{\mathrm{mul}} (b, k_b^{\mathrm{mul}_1},  \left \langle \hat{x} \right \rangle_b, \left \langle \hat{t} \right \rangle_b)$ ; \\
        $\left \langle \hat{t} \right \rangle_b = \mathrm{Eval}^{\mathrm{exp}}_{l}(b, k_b^{\mathrm{exp}}, -1 \cdot\left \langle \hat{x} \right \rangle_b) + 1$ ;\\
        $\left \langle t \right \rangle_b = \mathrm{Eval}^{\mathrm{recip}}_{l} (b, k_b^{\mathrm{recip}}, \left \langle \hat{t} \right \rangle_b)$\\
        $\left \langle \hat{t} \right \rangle_b = \mathrm{Eval}_{l, l}^{\mathrm{mul}} (b, k_b^{\mathrm{mul}_2}, \left\langle t \right\rangle_b + r_b^{(4)}, 1 - \left\langle p \right\rangle_b  + r_b^{(5)}) + \mathrm{Eval}_{l, l}^{\mathrm{mul}} (b, k_b^{\mathrm{mul}_3}, 1 -\left\langle t \right\rangle_b + r_b^{(6)}, \left\langle p \right\rangle_b  + r_b^{(7)}) + r_b^{\mathrm{out}}$ ;\\
        \textbf{return} $(\left \langle \hat{t} \right \rangle_0, \left \langle \hat{t} \right \rangle_1)$
    }
\end{algorithm}

\subsection{Dropout Function}
\label{appendix-dropout}
The algorithm first samples random numbers $r \in (0,1)$ for comparison (line 3).
The dropout factor $\sigma$ is determined by the comparison result $\mathbf{1}\{r < p\}$ and the dropout probability $p$ (line 4).
We split the offsets (line 5) and mask $\sigma$ by $r^{\sigma}$ (line 6) so that each party cannot learn $\sigma$ from keys.
Algorithm~\ref{alg-dropout-online} presents the online stage of the dropout function.
The main idea of the algorithm is to realize the multiplication of the dropout factor $\hat{\sigma}$ and the input $\hat{x}$ through $\mathrm{Eval}_{l,l}^{\mathrm{mul}}$.
Since $\hat{\sigma}$ is embedded into the keys during the offline stage, the evaluation omits the restoration of $\hat{\sigma}$ and reduces communication overhead.

\begin{algorithm}
    \caption{Dropout function with FSS in the offline stage $\mathrm{Gen}^{\mathrm{drop}}_{l} (r^{\text{in}}, r^{\text{out}}, r^{\sigma}, p)$}
    \label{alg-dropout-offline}
    \textbf{Requires:} Input offset $r^{\text{in}}$, output offset $r^{\text{out}}$, factor offset $r^{\sigma}$, dropout probability $p$. \\
    \textbf{Output:} Evaluation keys $(k_0, k_1)$. \\
    Sample random $r \in (0,1)$; \\
    $\sigma = \frac{1}{1-p} \cdot  \mathbf{1}\{r < p\}$ ; \\
    $\hat{\sigma} = \sigma + r^{\sigma}$;\\
    $k_b^{\mathrm{mul}} = \mathrm{Gen}_{l,l}^{\mathrm{mul}} (r^{\mathrm{in}}, r^{\mathrm{\sigma}}, r^{\mathrm{out}})$; \\
    For $b\in \{0,1\}$, let $k_b = k_b^{\mathrm{mul}} \| \hat{\sigma}$  ; \\
    \textbf{return} $(k_0, k_1)$
\end{algorithm}

\begin{algorithm}
    \caption{Dropout function with FSS in the online stage $\mathrm{Eval}^{\mathrm{drop}}_{l}(b, k_b, \left \langle \hat{x} \right \rangle_b)$}
    \label{alg-dropout-online}
    \textbf{Requires: } Evaluation keys $(k_0, k_1)$, additive shares $(\left \langle \hat{x} \right \rangle_0, \left \langle \hat{x} \right \rangle_1)$ of $\hat{x}$ with offset $r^{\text{in}}$, scale factor $s$.\\
    \textbf{Output:} Shares of $\mathrm{Dropout}(\hat{x} - x^{\mathrm{in}}) + r^{\text{out}}$.\\
    \For{$b \in \{0, 1\}$}{
        Parse $k_b = k_b^{\mathrm{mul}} \| \hat{\sigma}$;\\
        Parse $k_b^{\mathrm{mul}} = r^{\text{in}_1}_b \| r^{\text{in}_2}_b \| q_b \| r^{\text{out}}_b$;\\
        Restore $\hat{x} = x + r^{\mathrm{in}_1}$ from shares of $\hat{x}$ ;\\
        $t_b =  \mathrm{Trunc}(b \hat{x} \hat{\sigma} - \hat{x} r^{\mathrm{in}_2}_b - \hat{\sigma} r^{\mathrm{in}_1}_b + q_b, s)  + r^{\text{out}}_b \mod 2^l$; \\
    }
    \textbf{return} $(\left \langle \hat{t} \right \rangle_0, \left \langle \hat{t} \right \rangle_1)$
\end{algorithm}

\subsection{Tensor Product of Vectors}
\label{appendix-tensor-product}
Algorithms~\ref{alg-tensor-product-offline} and \ref{alg-tensor-product-online} provide the online and offline stages of the secure tensor product, respectively.
\begin{algorithm}
    \caption{Tensor product with FSS in the offline stage $\mathrm{Gen}^{\mathrm{TP}}_{N \times l, M \times l}
    \label{alg-tensor-product-offline}(\mathbf{r}^{\text{in}_1},\mathbf{r}^{\mathrm{in}_2}, \mathbf{r}^{\text{out}})$ .}
    \textbf{Requires: } Input offsets $\{ \mathbf{r}^{\text{in}_1}, \mathbf{r}^{\text{in}_2} \}$, output offset $\mathbf{r}^{\text{out}}$.\\
    \textbf{Output:} Evaluation keys $(k_0, k_1)$. \\
    Sample random $\mathbf{r}^{\text{in}_1}_0$, $\mathbf{r}^{\text{in}_1}_1 \in \mathbb{Z}_{2^l}^{N}$, $\mathbf{r}^{\text{in}_2}_0$, $\mathbf{r}^{\text{in}_2}_1 \in \mathbb{Z}_{2^l}^{M}$, $\mathbf{r}^{\text{out}}_0$, $\mathbf{r}^{\text{out}}_1 \in \mathbb{Z}_{2^l}^{N \times M}$, s.t., $\mathbf{r}^{\text{in}_1}_0 + \mathbf{r}^{\text{in}_1}_1 = \mathbf{r}^{\text{in}_1} \mod 2^l $, $\mathbf{r}^{\text{in}_2}_0 + \mathbf{r}^{\text{in}_2}_1 = \mathbf{r}^{\text{in}_2} \mod 2^l $, $\mathbf{r}^{\text{out}}_0 + \mathbf{r}^{\text{out}}_1 = \mathbf{r}^{\text{out}} \mod 2^l$; \\
    Sample random $\mathbf{q}_0, \mathbf{q}_1, \in \mathbb{Z}_{2^l}^{N \times M}$, s.t., $\mathbf{q}_0 + \mathbf{q}_1 = \mathbf{r}^{\text{in}_1} \otimes \mathbf{r}^{\text{in}_2} \mod 2^l$; \\
    For $b\in \{0,1\}$, let $k_b = \mathbf{r}^{\text{in}_1}_b \| \mathbf{r}^{\text{in}_2}_b \| \mathbf{q}_b \| \mathbf{r}^{\text{out}}_b$ ;\\
    \textbf{return} $(k_0, k_1)$
\end{algorithm}

\begin{algorithm}
    \caption{Tensor product with FSS in the online stage $\mathrm{Eval}^{\mathrm{TP}}_{N \times  l,M \times l}
    (\mathbf{r}^{\text{in}_1},\mathbf{r}^{\mathrm{in}_2}, \mathbf{r}^{\text{out}})$.}
    \label{alg-tensor-product-online}
    \textbf{Requires: } Evaluation keys $(k_0, k_1)$, additive shares $(\left \langle \hat{\mathbf{x}} \right \rangle_0, \left \langle \hat{\mathbf{x}} \right \rangle_1)$ of $\hat{x}$ with offset $\mathbf{r}^{\text{in}_1}$, additive shares $(\left \langle \hat{\mathbf{y}} \right \rangle_0, \left \langle \hat{\mathbf{y}} \right \rangle_1)$ of $\hat{\mathbf{y}}$ with offset $\mathbf{r}^{\text{in}_2}$, scale factor $s$.\\
    \textbf{Output:} Shares of $(\hat{\mathbf{x}} - \mathbf{r}^{\text{in}_1}) \otimes (\hat{\mathbf{y}} - \mathbf{r}^{\text{in}_2}) + \mathbf{r}^{\text{out}}$.\\
    \For{$b \in \{0, 1\}$}{
        Parse $k_b = \mathbf{r}^{\text{in}_1}_b \| \mathbf{r}^{\text{in}_2}_b \| \mathbf{q}_b \| \mathbf{r}^{\text{out}}_b$;\\
        Restore $\hat{\mathbf{x}} = \mathbf{x} + r^{\mathrm{in}_1}$ from shares of $\hat{\mathbf{x}}$ ;\\
        Restore $\hat{\mathbf{y}} = \mathbf{y} + r^{\mathrm{in}_2}$ from shares of $\hat{\mathbf{y}}$ ;\\
        $\left \langle \hat{t} \right \rangle_b = b\cdot (\hat{\mathrm{x}} \otimes \hat{\mathrm{y}}) + \hat{\mathrm{x}} \otimes \mathbf{r}^{\text{in}_2}_b + \mathbf{r}^{\text{in}_1}_b \otimes \hat{\mathrm{y}} + \mathbf{q}_b$;\\
        $ \left \langle \hat{t} \right \rangle_b = \mathrm{Trunc}(\left \langle \hat{t} \right \rangle_b ,s) + \mathbf{r}^{\text{out}}_b$ ;\\ 
    }
    \textbf{return} $(\left \langle \hat{t} \right \rangle_0, \left \langle \hat{t} \right \rangle_1)$
\end{algorithm}

\section{Security Analysis}
\label{appendix-security-analysis}
We provide proof of semi-honest simulation-based security~\cite{lindell2017simulate} for the proposed protocols.
In scenarios where a protocol integrates a sub-protocol to perform a particular function $\mathcal{F}$, we tackle the security verification by replacing the sub-protocol invocation with direct access to the function $\mathcal{F}$ itself. 
This approach is known as the $\mathcal{F}$-hybrid model, facilitating a modular security proof that presumes the security of the fundamental function $\mathcal{F}$. 
Consequently, we can concentrate on validating the security of the higher-level protocol, leveraging the confirmed security of the sub-protocol that is embedded within the function $\mathcal{F}$.

\begin{lemma}
    \label{lemma-security-mul}
    Protocols $\mathrm{Gen}^{\mathrm{mul}}_{l,l}$ and $\mathrm{Eval}^{\mathrm{mul}}_{l,l}$ in Algorithms \ref{alg-mul-offline} and \ref{alg-mul-online} securely realize $\mathcal{F}_{\mathrm{mul}}$.
\end{lemma}
\begin{proof}
    We prove the security of $\mathcal{F}_{\mathrm{mul}}$ under the semi-honest model. In the offline stage, party $P_b$ learns $r^{\mathrm{in}_1}_b, r^{\mathrm{in}_2}_b, q_b, r^{\mathrm{out}}_b \in \mathbb{Z}_{2^{l}}$. All of the above values are random. Therefore, the information learned by $P_b$ in the offline stage can be perfectly simulated. The only information that each party $P_b$ learns in the online stage is $\hat{x}$ and $\hat{y}$, which is masked by the random values $r^{\mathrm{in}_1}, r^{\mathrm{in}_2} \in \mathbb{Z}_{2^{l}}$. Hence, the distribution of $\hat{x}$ and $\hat{y}$ is uniformly random from the view of $P_b$. Therefore, the information learned by $P_b$ can be perfectly simulated.
    Based on the composition theorem~\cite{canetti2000security}, we can claim that $\mathrm{Gen}^{\mathrm{mul}}_{l,l}$ and $\mathrm{Eval}^{\mathrm{mul}}_{l,l}$ securely realize $\mathcal{F}_{\mathrm{mul}}$.
\end{proof}

The security proof of $\mathcal{F}_{\mathrm{square}}$ and $\mathcal{F}_{\mathrm{TP}}$ given in Algorithms~\ref{alg-square-offline}, \ref{alg-square-online}, \ref{alg-tensor-product-offline}, and \ref{alg-tensor-product-online} is similar to the proof of Lemma~\ref{lemma-security-mul}.

\begin{lemma}
    \label{lemma-security-power}
    Protocols $\mathrm{Gen}^{\mathrm{power}}_{l}$ and $\mathrm{Eval}^{\mathrm{power}}_{l}$ in Algorithms~\ref{alg-power-offline} and \ref{alg-power-online} securely realize $\mathcal{F}_{\mathrm{power}}$ in the $(\mathcal{F}_{\mathrm{mul}}, \mathcal{F}_{\mathrm{square}})$-hybrid model.
\end{lemma}
\begin{proof}
    We prove the security of $\mathcal{F}_{\mathrm{power}}$ under the semi-honest model. We first discuss the case of $n=2^n$. In the offline stage, party $P_b$ learns $k_b$, i.e., the keys generated by $m$ times $\mathrm{Gen}_{l}^{\mathrm{square}}$. Similarly, in the online stage, party $P_b$ learns masked values $\hat{x}$ in $\mathrm{Eval}_{l}^{\mathrm{square}}$. Therefore, we can claim that the proposed protocols realize secure $\mathcal{F}_{\mathrm{power}}$ in the $\mathcal{F}_{\mathrm{square}}$-hybrid model. The power function is implemented according to Equation~\eqref{equation-power-function} when $n\neq2^n$. The implementation invokes the secure multiplication and square function. Since we have proved the security of the above functions, based on the composition theorem~\cite{canetti2000security}, we can claim that $\mathrm{Gen}^{\mathrm{power}}_{l}$ and $\mathrm{Eval}^{\mathrm{power}}_{l}$ in in Algorithms~\ref{alg-power-offline} and \ref{alg-mul-online} securely realize $\mathcal{F}_{\mathrm{power}}$ in the $(\mathcal{F}_{\mathrm{mul}}, \mathcal{F}_{\mathrm{square}})$-hybrid model.
\end{proof}

\begin{lemma}
    \label{lemma-security-exp}
    Protocols $\mathrm{Gen}^{\exp}_{l}$ and $\mathrm{Eval}^{\exp}_{l}$ in Algorithms~\ref{alg-exp-offline} and \ref{alg-exp-online} securely realize $\mathcal{F}_{\exp}$ in the $\mathcal{F}_{\mathrm{square}}$-hybrid model.
\end{lemma}
\begin{proof}
    We prove the security of $\mathcal{F}_{\mathrm{exp}}$ under the semi-honest model. Party $P_b$ learns the key of $k_b$ in the offline stage where $k_b$ contains random shares of $r^{\mathrm{in}}$, $r^{\mathrm{out}}$, and the keys from the secure square function $\{r_b^{(1)},q_b^{(1)},\cdots, r_b^{(m_{\exp})}, q_b^{(m_{\exp})}\}$. According to the functionality of $\mathrm{Gen}^{\mathrm{square}}_{l}$, elements in $\{r_b^{(1)},q_b^{(1)},\cdots, r_b^{(m_{\exp})}, q_b^{(m_{\exp})}\}$ are random values from $\mathbb{Z}_{2^l}$.
    In the online stage, party $P_b$ implements $m_{\exp}$ times $\mathrm{Eval}^{\mathrm{square}}_{l}$ and learns the restoration of masked input $\hat{x}$ in each $\mathrm{Eval}^{\mathrm{square}}_{l}$. According to the functionality of $\mathrm{Gen}^{\mathrm{square}}_{l}$, input $\hat{x}$ is masked by random values. Therefore, the distributions of the above information that $P_b$ can learn are uniformly random from the view of $P_b$. Based on the composition theorem~\cite{canetti2000security}, we claim that $\mathrm{Gen}^{\exp}_{l}$ and $\mathrm{Eval}^{\exp}_{l}$ in Algorithms~\ref{alg-exp-offline} and \ref{alg-exp-online} securely realize $\mathcal{F}_{\exp}$ in the $\mathcal{F}_{\mathrm{square}}$-hybrid model.
\end{proof}

\begin{lemma}
    \label{lemma-security-recip}
    Protocols $\mathrm{Gen}^{\mathrm{recip}}_{l}$ and $\mathrm{Eval}^{\mathrm{recip}}_{l}$ in Algorithms~\ref{alg-recip-offline} and \ref{alg-recip-online} securely realize $\mathcal{F}_{\mathrm{recip}}$ in the $(\mathcal{F}_{\exp}, \mathcal{F}_{\mathrm{power}})$-hybrid model.
\end{lemma}
\begin{proof}
    We prove the security of $\mathcal{F}_{\mathrm{recip}}$ under the semi-honest model. According to Algorithm~\ref{alg-recip-offline}, party $P_b$ learns the $\exp$ keys $k_b^{\exp}$ and shared values $\{r_b^{\mathrm{in}}, r_b^{\mathrm{out}}, r_b^{(1)}, q_b^{(1)}, u_b^{(1)}, v_b^{(1)}, \cdots, v_b^{(m_{\mathrm{recip}})}\}$ in the offline stage. According to the functionality of $\mathcal{F}_{\exp}$, $k_b^{\exp}$ are random values. According to the functionality of $\mathcal{F}_{\mathrm{power}}$, $\{r_b^{\mathrm{in}}, r_b^{\mathrm{out}}, r_b^{(1)}, q_b^{(1)}, u_b^{(1)}, v_b^{(1)}, \cdots, v_b^{(m_{\mathrm{recip}})}\}$ are random values from $\mathbb{Z}_{2^l}$. In the online stage, party $P_b$ learns the following information: $\hat{x}$, the outputs of $\mathrm{Eval}^{\exp}_l$, $ \mathrm{Eval}^{\mathrm{square}}_l$, and $\mathrm{Eval}^{\mathrm{mul}}_l$. Specifically, $\hat{x}$ is masked by a random value $r^{\mathrm{in}} \in \mathbb{Z}_{2^l}$. According to the functionally of $\mathcal{F}_{\exp}$, the output of $\mathrm{Eval}^{\exp}_l$ is masked by a random value from $\mathbb{Z}_{2^l}$. According to the functionally of $\mathcal{F}_{\mathrm{power}}$, the outputs of $ \mathrm{Eval}^{\mathrm{square}}_l$ and $\mathrm{Eval}^{\mathrm{mul}}_l$ are masked by random values from $\mathbb{Z}_{2^l}$. Therefore, the distributions of the above information that $P_b$ can learn are uniformly random from the view of $P_b$. Based on the composition theorem~\cite{canetti2000security}, we claim that $\mathrm{Gen}^{\mathrm{recip}}_{l}$ and $\mathrm{Eval}^{\mathrm{recip}}_{l}$ in Algorithms~\ref{alg-recip-offline} and \ref{alg-recip-online} securely realize $\mathcal{F}_{\mathrm{recip}}$ in the $(\mathcal{F}_{\exp}, \mathcal{F}_{\mathrm{power}})$-hybrid model.
\end{proof}

The security proof of $\mathcal{F}_{\mathrm{dropout}}$ given in Algorithms~\ref{alg-dropout-offline} and \ref{alg-dropout-online} is similar to the proof of Lemma~\ref{lemma-security-recip}.

\begin{lemma}
    \label{lemma-security-softmax}
    Protocols $\mathrm{Gen}^{\mathrm{softmax}}_{l}$ and $\mathrm{Eval}^{\mathrm{softmax}}_{l}$ in Algorithms~\ref{alg-softmax-offline} and \ref{alg-softmax-online} securely realize $\mathcal{F}_{\mathrm{softmax}}$ in the $\mathcal{F}_{\mathrm{recip}}$-hybrid model.
\end{lemma}
\begin{proof}
    We prove the security of $\mathcal{F}_{\mathrm{softmax}}$ under the semi-honest model. According to Algorithm~\ref{alg-softmax-offline}, party $P_b$ learns random shared values $r_b^{\mathrm{in}}, r_b^{\mathrm{out}}$, and the multiplication function keys in the offline stage. According to the functionally of $\mathcal{F}_{\mathrm{recip}}$, multiplication keys are random values. According to the Algorithm~\ref{alg-softmax-online}, party $P_b$ learns the output of $\mathrm{ReLU}$ and $\mathrm{Eval}^{\mathrm{mul}}_{K \times l, K \times l}$. The security of $\mathrm{ReLU}$ is given by CrypTen~\cite{Crypten}. According to the functionally of $\mathcal{F}_{\mathrm{recip}}$, the outputs of $\mathrm{Eval}^{\mathrm{mul}}_{K \times l, K \times l}$ are random values from $\mathbb{Z}_{2^l}^{K}$. Therefore, the distributions of the above information that $P_b$ can learn are uniformly random from the view of $P_b$. Based on the composition theorem~\cite{canetti2000security}, we claim that $\mathrm{Gen}^{\mathrm{softmax}}_{l}$ and $\mathrm{Eval}^{\mathrm{softmax}}_{l}$ in Algorithms~\ref{alg-softmax-offline} and \ref{alg-softmax-online} securely realize $\mathcal{F}_{\mathrm{softmax}}$ in the $ \mathcal{F}_{\mathrm{recip}}$-hybrid model.
\end{proof}

The security proof of $\mathcal{F}_{\mathrm{sigmoid}}$ and $\mathcal{F}_{\mathrm{sigmoid}}$ given in Algorithms~\ref{alg-sigmoid-offline} and \ref{alg-sigmoid-online} is similar to the proof of Lemma~\ref{lemma-security-softmax}.

\section{Consumption of Fine-tuning Models}
\label{appendix-evaluation}
We provide the communication consumption and execution time of privacy-preserving fine-tuning of BERT through all training samples of downstream tasks in this section.
\begin{table}
	\caption{Communication consumption of privacy-preserving fine-tuning for BERT with the whole training data of downstream tasks.}
	\centering
	\label{table-evaluation-datasets-communication}
	\begin{tblr}{hline{3,4,6,7}={dotted}, vline{2}={2-Z}{}, vline{3}={2-Z}{}, column{1}={0.7cm}, column{2}={1.2cm}, column{3-6}={1cm}, columns={c,m}
		}
            \toprule
		  \SetCell[c=2]{c,m} Datasets   &     & SST-2    & MRPC     & RTE      & CoLA \\
            \midrule
            \SetCell[r=3]{c,m} BERT-base  & PriFFT-IT  & 383.74GB & 20.80GB  & 14.05GB  & 48.70GB \\
                                          & PriFFT-LT  & 70.11GB  & 3.80GB   & 2.57GB   & 8.90GB \\
                                          & ABY2 & 1.59TB   & 86.99GB  & 58.76GB  & 203.73GB \\
            \midrule
            \SetCell[r=3]{c,m} BERT-large & PriFFT-IT  & 655.86GB & 35.54GB  & 24.01GB  & 83.23GB \\
                                          & PriFFT-LT  & 106.20GB & 5.76GB   & 3.90GB   & 13.48GB \\
                                          & ABY2 & 2.78TB   & 152.30GB & 102.87GB & 356.77GB  \\
            \bottomrule
	\end{tblr}
\end{table}

\begin{table}
	\caption{Execution time of privacy-preserving fine-tuning for BERT with the whole training data of downstream tasks.}
	\centering
	\label{table-evaluation-datasets-time}
	\begin{tblr}{hline{5,11}, hline{3,4,6,7,9,10,12,13}={dotted}, vline{2}={2-Z}{}, vline{3}={2-Z}{}, column{1}={0.55cm}, column{2}={0.66cm}, column{3}={1.3cm}, column{4-7}={0.81cm}, columns={c,m}
		}
            \toprule
		  \SetCell[c=3]{c,m} Datasets & & & SST-2 & MRPC & RTE & CoLA \\
            \midrule
		\SetCell[r=6]{c,m} CPU & \SetCell[r=3]{c,m} BERT-base   & PriFFT-IT  & 84.86m  & 4.60m  & 3.11m & 10.77m \\
                                   &                          & PriFFT-LT  & 40.03m  & 1.87m  & 1.30m & 4.21m \\
                                   &                          & ABY2 & 150.09m & 8.13m  & 5.92m & 19.61m \\
                                   & \SetCell[r=3]{c,m} BERT-large  & RL  & 131.85m & 6.61m  & 5.41m & 15.16m \\
                                   &                          & PriFFT-LT  & 35.12m  & 2.19m  & 1.88m & 4.56m \\
                                   &                          & ABY2 & 238.38m & 12.32m & 8.73m & 30.83m \\
            \midrule
            \SetCell[r=6]{c,m} GPU & \SetCell[r=3]{c,m} BERT-base   & PriFFT-IT  & 45.42m  & 2.08m  & 1.48m & 5.48m \\
                                   &                          & PriFFT-LT  & 19.35m  & 1.02m  & 44s   & 2.33m \\
                                   &                          & ABY2 & 110.46m  & 5.77m  & 4.23m & 12.89m \\
                                   & \SetCell[r=3]{c,m} BERT-large  & PriFFT-IT  & 72.5m   & 3.87m  & 2.41m & 9.48m \\
                                   &                          & PriFFT-LT  & 23.35m  & 1.60m  & 56s   & 3.13m \\
                                   &                          & ABY2 & 185.50m  & 10.05m  & 6.79m & 23.54m \\
            \bottomrule
	\end{tblr}
\end{table}

The batch size is set to 32, and the number of training samples for each task is: SST-2 contains 67.3K samples (2104 batches); MRPC contains 3.67K samples (114 batches); RTE contains 2.49K samples (77 batches) ; CoLA contains 8.55K samples (267 batches).
The evaluation considers one client and the server to perform privacy-preserving fine-tuning on the above training data to analyze the communication overhead and execution time.
The results are summarized in TABLE~\ref{table-evaluation-datasets-communication} and TABLE~\ref{table-evaluation-datasets-time}.
For the same implementation setting, communication overhead and execution time is linearly related to the amount of training data in the downstream tasks.
PriFFT-IT and PriFFT-LT apply the optimized protocols proposed in Section~\ref{section-mechanism} for privacy-preserving fine-tuning, resulting in lower resource consumption for both implementations than fine-tuning with ABY2.
Meanwhile, PriFFT-LT skips the iterative truncations and uses the local truncation, reducing communication and computation.
Therefore, we apply PriFFT-LT in the initial stage to minimize resource consumption and obtain more accurate model parameters through PriFFT-IT when the model tends to converge.

\end{appendices}

\end{document}